\documentclass[journal]{IEEEtran}

\usepackage{cite}
\usepackage[pdftex]{graphicx}
\usepackage{amsmath}
\usepackage{algorithmic}
\usepackage{array}
\usepackage{fixltx2e}
\usepackage{url}
\usepackage{amsfonts}
\usepackage{algorithm}
\usepackage{algorithmic}
\usepackage{bm}
\usepackage{arydshln}
\usepackage{makecell,multirow,diagbox} 
\usepackage{booktabs}
\usepackage{times}  
\usepackage{helvet}  
\usepackage{courier}  
\usepackage{subeqnarray}
\usepackage{cases}
\newtheorem{theorem}{Theorem}

\newtheorem{proposition}{Proposition}
\newtheorem{proof}{Proof}[section]
\newtheorem{remark}{Remark}
\begin{document}
	\title{Converged Deep Framework Assembling Principled Modules for CS-MRI}
	
	\author{Risheng Liu,~\IEEEmembership{Member,~IEEE,}
		Yuxi~Zhang,
		Shichao~Cheng,
		Zhongxuan~Luo,
		and~Xin~Fan,~\IEEEmembership{Senior~Member,~IEEE}
		\thanks{This work was supported in part by the National Natural Science Foundation of China (NSFC) under Grant Nos 61922019, 61672125, 61572096 and 61632019, the Youth Tiptop Talent project in Liaoning Province (XLYC1807088) and the Fundamental Research Funds for the Central Universities. \emph{(Corresponding author: Xin Fan).}}
		\thanks{R. Liu, Y. Zhang, Z. Luo and X. Fan are with the DUT-RU International School of Information Science and Technology, Dalian University of Technology, Dalian, 116620, China, and also with the Liaoning Key Laboratory of Ubiquitous Network and Service Software (e-mail: rsliu@dlut.edu.cn; yuxizhang@mail.dlut.edu.cn; zxluo@dlut.edu.cn; xin.fan@ieee.org).}
		\thanks{S. Cheng is with the School of Computer Science and Technology, Hangzhou Dianzi University, Hangzhou, 310018, China (e-mail: shichao.cheng@outlook.com).}
		\thanks{Z. Luo is also with the Institute of Artificial Intelligence, Guilin University of Electronic Technology, Guilin, China}
		}
	
	\markboth{Journal of \LaTeX\ Class Files,~Vol.~14, No.~8, August~2015}%
	{Shell \MakeLowercase{\textit{et al.}}: Bare Demo of IEEEtran.cls for IEEE Journals}

	\maketitle
	
	\begin{abstract}
	Compressed Sensing Magnetic Resonance Imaging (CS-MRI) significantly accelerates MR data acquisition at a sampling rate much lower than the Nyquist criterion. A major challenge for CS-MRI lies in solving the severely ill-posed inverse problem to reconstruct aliasing-free MR images from the sparse k-space data. Conventional methods typically optimize an energy function, producing reconstruction of high quality, but their iterative numerical solvers unavoidably bring extremely slow processing. Recent data-driven techniques are able to provide fast restoration by either learning direct prediction to final reconstruction or plugging learned modules into the energy optimizer. Nevertheless, these data-driven predictors cannot guarantee the reconstruction following constraints underlying the regularizers of conventional methods so that the reliability of their reconstruction results are questionable. In this paper, we propose a converged deep framework assembling principled modules for CS-MRI that fuses learning strategy with the iterative solver of a conventional reconstruction energy. This framework embeds an optimal condition checking mechanism, fostering \emph{efficient} and \emph{reliable} reconstruction. We also apply the framework to two practical tasks, \emph{i.e.}, parallel imaging and reconstruction with Rician noise. Extensive experiments on both benchmark and manufacturer-testing images demonstrate that the proposed method reliably converges to the optimal solution more efficiently and accurately than the state-of-the-art in various scenarios. 
	\end{abstract}
	\begin{IEEEkeywords}
		Compressed sensing, MRI reconstruction, Deep learning, Optimization, Parallel imaging, Rician noise
	\end{IEEEkeywords}
	
	\IEEEpeerreviewmaketitle
	
	\section{Introduction}

	\IEEEPARstart{M}{agnetic} Resonance Imaging (MRI) is a widely applied none-invasive technology that reveals both functional and anatomical information. Unfortunately, traditional MRI  suffers from inherently slow acquisition from a large volume of data in the k-space, \emph{i.e.}, Fourier space~\cite{lustig2008compressed}. Researchers introduce the compressed sensing (CS) theory into MRI reconstruction, allowing fast acquisition at a sampling rate much lower than the Nyquist rate without degrading image quality. The main challenge for CS-MRI is the illness of the inherited inverse problem to estimate images of high quality from under-sampled data points. Moreover, the existence of acquisition noise and computation of discrete Fourier transformation deteriorate the reconstruction stability.

	Conventional approaches to CS-MRI devote significant efforts to incorporating prior knowledge on images as regularizers in order to restore details as well as suppress artifacts. Sparsity constraints in a specific transform domain are the most commonly used~\cite{block2007undersampled,qu2012undersled,lustig2007sparse,gho2010three}. Dictionary learning based methods consider the sparsity in the subspace spanned by reference images~\cite{ravishankar2011mr,babacan2011reference,zhan2016fast}. Non-local paradigms take the advantage of the sparse representation by grouping similar local image patches~\cite{qu2014magnetic,eksioglu2016decoupled}. These methods derived from certain prior models are able to find a stable and satisfactory reconstruction by optimizing an energy functional with various forms of regularizers. Nevertheless, this optimization typically demands iterative calculation of gradient information in a high-dimensional image space, suffering from rather slow process. 
	
	Recent deep learning based approaches directly apply pre-trained deep architectures to infer MR images from degraded sampled inputs~\cite{wang2016accelerating,lee2017deep,schlemper2018deep}, rendering fast reconstruction. Nevertheless, these methods upon direct deep mappings are not so stable as conventional ones to data variations because they abandon the principled prior knowledge and totally depend on training examples. By taking advantage of both domain knowledge and available data, hybrid paradigms attempt to bridge the gap between the data learning and numerical optimization for functional energy with regularizers. For instance, Sun~\emph{et~al.}~\cite{sun2016deep} and Diamond~\emph{et~al.}~\cite{diamond2017unrolled} unroll the iterative optimization as the prediction process using deep networks. Unfortunately, this type of \emph{ad~hoc} combinations breaks the convergence of optimization iterations so that it cannot guarantee the optimal solution to the original objective energy. Figure~\ref{fig:CmpWithNet} illustrates an example of reconstruction from the latest deep learning based method~\cite{quan2018compressed} where the circled anatomic structures mistakenly appear artifacts while smearing part of tissues, though increasing PSNR from $26.25$ to $31.74$. Therefore, neither medical research nor clinical diagnosis can trust reconstruction without theoretical guarantee.   
	
	Existing methods fail to provide an \emph{efficient} or \emph{reliable} treatment for CS-MRI. In this study, we design a converged deep framework comprising modules assembled in principle for trustworthy fast MR reconstruction. Specifically, we provide an efficient iterative process by integrating the task-driven domain knowledge with the data-driven deep networks. Meanwhile, we introduce an automatic feedback mechanism that leads deep iterations \emph{converging} to the optimal solution of the original energy. Further, we consider two extensive tasks in practical scenarios,~\emph{i.e.}, parallel imaging and Rician noise pollution, to investigate the robustness of our paradigm for real-world applications. Our framework flexibly adapts to the energies targeting to these scenarios, and robustly finds the solutions. Our contributions can be listed as follows:
	\begin{itemize}
		\item We propose a generic deep framework assembling modules in principle to efficiently and reliably solve the CS-MRI model, typically non-convex and non-smooth.
		\item We establish an automatic feedback mechanism in virtue of the first-order optimality condition to reject improper predictions from embedded deep architectures and to guide the deep iterations towards the desired solution. 
		\item We custom-tailor two extensions of our framework to handle the case of parallel imaging and of Rician noise contamination in practical scenarios. 
	\end{itemize}
	
	Extensive experimental results on benchmarks and raw data collected in the $k$-space of an MRI equipment demonstrate the superiority of the proposed paradigm against state-of-the-art techniques in both reconstruction accuracy and efficiency, as well as the robustness to noise contamination.
		\begin{figure}[!tbp]
		\begin{center}
			\begin{tabular}{c@{\extracolsep{-1em}}c@{\extracolsep{0.2em}}c@{\extracolsep{0.2em}}c@{\extracolsep{0.2em}}c}
				&\includegraphics[width=.12\textwidth]{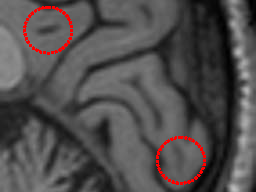}
				&\includegraphics[width=.12\textwidth]{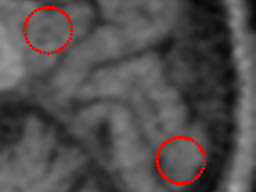}
				&\includegraphics[width=.12\textwidth]{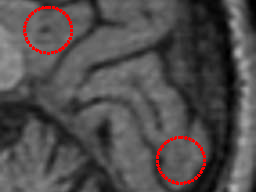}
				&\includegraphics[width=.12\textwidth]{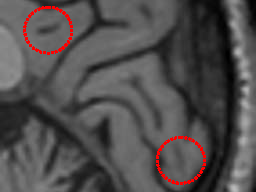}\\
				&Ground Truth & Zero-Filling & RefineGAN & Ours \\
			\end{tabular}
		\end{center}
		\caption{A reconstruction example from naive zero-filling, recent RefineGAN, and ours. Undesired artifacts and smoothing details in red circles emerge in the result of RefineGAN.    }
		\label{fig:CmpWithNet}
	\end{figure}
\begin{figure*}[!tbp]
	\begin{center}
		\begin{tabular}{c@{\extracolsep{-1em}}c@{\extracolsep{0.2em}}c}
			&\includegraphics[width=1\textwidth]{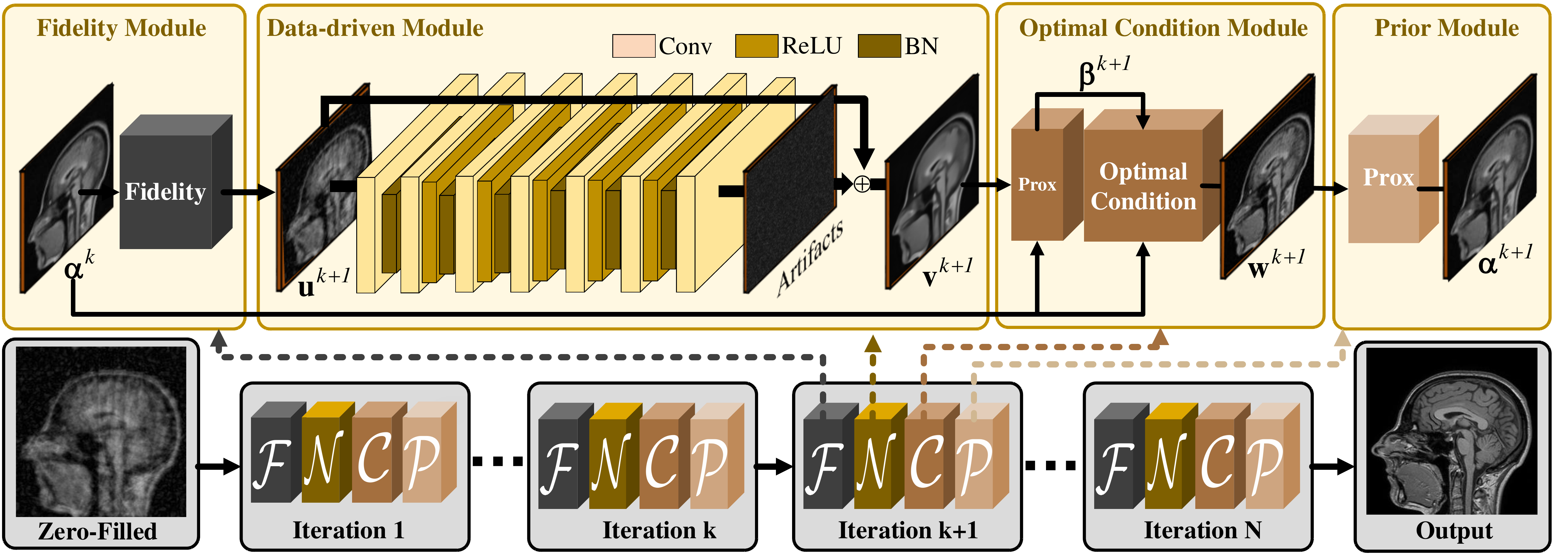}
		\end{tabular}
	\end{center}
	\caption{The fidelity $\mathcal{F}$, data-driven $\mathcal{N}$, optimal condition $\mathcal{C}$, and prior $\mathcal{P}$ modules forms one iteration of our reconstruction framework. }
	\label{fig:FlowChart}
\end{figure*}
	\section{Related Work}
	This section reviews approaches to CS-MRI reconstruction based on prior models and deep learning, followed by the algorithms for two practical reconstruction scenarios, \emph{i.e.}, parallel imaging and Rician noise removal. 
	\subsection{Model based CS-MRI}
	Conventional CS-MRI needs to solve an energy optimization with regularizations derived from image priors in a specific transform domain or subspace. This term plays a vital role in the reconstruction quality. Researchers have devoted tremendous efforts to exploring the sparsity prior in various transform domains including the gradient~\cite{block2007undersampled}, discrete Cosine transform~\cite{lustig2007sparse}, partial Fourier transform (RecPF)~\cite{yang2010fast}, and wavelet transform domains~\cite{qu2012undersled}. These techniques, developed upon the sparse representation with a fixed set of basis functions, can hardly provide satisfactory quality. Alternatively, researchers resort to the data-driven sparse prior in the subspace spanned by an over-complete dictionary of reference images, performing better in details recovery and artifacts removal~\cite{ravishankar2011mr,babacan2011reference,zhan2016fast,qu2014magnetic,eksioglu2016decoupled}. These dictionary based CS-MRI methods are able to generate reliable or plausible reconstruction following the underlying physical prior or constraint, but typically suffer from extremely slow process as the energy optimization turns out to be a computationally expensive iterative process. 
	
	Additionally, most of regularization terms derive the sparsity prior in the forms of the $l_p(0\le p \le 1)$ norm. The norm has nice mathematical properties, but may not be able to characterize complex data distributions in real applications. Also, simply concatenating reference images into an over-complete dictionary as the data dependent prior has the limitation in handling flexible structures of real imaging problems.
	
	\subsection{Deep learning based CS-MRI}
	Recent studies introduce deep learning techniques into MRI reconstruction in order to utilize the information given by vaults of available data examples. Wang~\emph{et~al.} design and train a convolutional neural network (CNN) to learn the direct mapping from the input under-sampled $k$-space data to the desired fully-sampled image~\cite{wang2016accelerating}. A modified U-Net architecture is used to learn the aliasing artifacts in \cite{lee2017deep} and a deep cascading of convolutional network is applied to accelerate MR imaging in\cite{schlemper2018deep}.Lately, generative adversarial networks (GANs) are able to remove aliasing artifacts by generated details, and gain the state-of-the-art quality for a specific benchmark~\cite{quan2018compressed,yang2017dagan,mardani2018deep}. These methods provide an extremely fast feed-forward inference for reconstruction, but highly depend on training examples without any principled knowledge, failing to give the desired solution when the input deviates from the mode of training data distribution.

	There exist hybrid techniques that integrate deep architectures into an iterative optimization process for the reconstruction energy with prior knowledge. Sun~\emph{et~al.} present ADMM-Net for CS-MRI by introducing learnable architectures into the Alternating Direction Method of Multipliers (ADMM) iterations~\cite{sun2016deep}. A similar unrolling scheme is developed in~\cite{diamond2017unrolled}. These methods upon \emph{deep} priors can improve the reconstruction accuracy and reliability at a relatively faster speed. Unfortunately, the unrolling schemes that directly replace iterations by deep architectures cannot guarantee the convergence of the iterative optimization process. Moreover, no mechanism is available to control the errors generated by nested structures. Therefore, we can still hardly rely on these reconstructed results in medical analysis and diagnosis.
	 	
	\subsection{Parallel imaging algorithms}
	
	Parallel imaging (PI) techniques, similar with CS-MRI, are also commonly used in clinic applications to accelerate MRI acquisition by using different sensitivities received from multiple coils installed at different locations. For the sake of saving time, PI conducts spatial encoding in the form of receiver coils with various sensitivity profiles instead of gradient encoding. SENSE and its variants directly restore the images from under-sampled data in the image domain~\cite{pruessmann1999sense,ying2007joint}. Another line of studies is to first interpolate the missing data in the frequency domain, and then to transfer the results to the image domain, \emph{e.g.}, GRAPPA\cite{griswold2005parallel}, SPIRiT~\cite{lustig2010spirit}, and their derived techniques. Recent approaches combine a CS energy model with PI in order to further speed up the acquisition such as CS-SENSE~\cite{liang2009accelerating}, CS-GRAPPA~\cite{chang2012nonlinear}, $L_1$-SPIRiT~\cite{murphy2012fast}, and SAKE~\cite{shin2014calibrationless}. It is worth investigating how these energy-based PI approaches may benefit from deep learning.	
	\subsection{Rician noise removal}
	Independent zero-mean Gaussian noise with equal variances may simultaneously corrupt both real and imaginary parts of the complex data in the $k$-space during practical MR acquisition~\cite{rajan2012adaptive}. Typical MR imaging techniques convert the original complex Gaussian noise to Rician noise when computing the magnitude of the complex $k$-space data~\cite{gudbjartsson1995rician}. The non-additive Rician noise challenges traditional noise removal techniques. Researchers designate special treatments targeting at Rician noise removal, \emph{e.g.}, nonlocal-means algorithms\cite{wiest2008rician,manjon2010adaptive}, wavelet packet\cite{wood1999wavelet}, wavelet domain filtering\cite{nowak1999wavelet}, and total variation(TV) denoisers~\cite{chen2015convex,liu2016variational}. Deep learning based denoisers are also developed to tackle Rician noise~\cite{sun2016deep,yang2017dagan}. Nevertheless, it is still an open problem to fuse Rician noise removal with efficient CS-MRI reconstruction in a unified framework for practical MR acquisition. 
	

	\section{Method}
	We first give the problem formulation of CS-MRI, and present our deep framework. Subsequently, we detail the reconstruction process that converges to the optimum of the original energy, and finally provide the theoretical analysis.  
	\subsection{Problem formulation}
	According to the CS theory~\cite{eksioglu2016decoupled}, typical CS-MRI reconstruction attempts to recover a fully-sampled image $\mathbf{x}$ from under-sampled $k$-space data $\mathbf{y}$ as:
	\begin{equation}
	\min\limits_{\mathbf{x}}{ \bm{\Phi}(\mathbf{x})}\quad s.t.\quad \mathbf{PFx=y},
	\label{eq:AbstractModel}
	\end{equation}
	where $\mathbf{y} \in\mathbb{C}^M $ represents the under-sampled observation in the $k$-space, and $ \mathbf{x} \in\mathbb{C}^N (M<<N)$ denotes the column vector of the desired MR image to be reconstructed. The Fourier transform and under-sampling operation are denoted as $ \mathbf{F} $ and $ \mathbf{P} \in \mathbb{R}^{M\times N}$, respectively. The composite operator $\mathbf{PF}$ constitutes a linear operation $ \mathbb{C}^N \to \mathbb{C}^M $. $\bm{\Phi}(\mathbf{x})$ represents the regularization imposed on the ill-posed inverse problem for reconstruction in order to constrain the searching space of the desired solution $\mathbf{x}$. 
	
	Existing CS-MRI techniques usually solve the energy optimization~(\ref{eq:AbstractModel}) by iterations upon gradient information with the prior penalty. These methods suffer from expensive time consumption owing to the iterative gradient calculations. Alternatively, an end-to-end inference upon deep networks can achieve efficient MR reconstruction from under-sampled data. These direct approaches without explicit constraints from prior knowledge lack reliability or robustness to data variations. Hybrid strategies integrate both principled knowledge and deep architectures into the optimization procedure, but this simple combination has no theoretical guarantee on the convergence so that the reliability issue has not been resolved yet. Such trade of reliability for efficiency may produce serious consequence in medical analysis and diagnosis because one cannot tell whether a reconstruction algorithm restores authentic anatomical structures or fabricates/hallucinates details upon an available data distribution. In this work, we propose a deep reconstruction framework with theoretical guarantee on convergence, enabling both efficiency and robustness.


	
	\subsection{Deep optimization framework}
	We consider the classical CS-MRI model with the sparsity regularization on the wavelet basis:		
	\begin{equation}
	\bm{\alpha}\in
	\arg\min\limits_{\bm{\alpha}}
	\ \frac{1}{2}\|\mathbf{PFA}\bm{\alpha}\!-\!\mathbf{y}\|_2^2\!+
	\!\lambda\|\bm{\alpha}\|_p,
	\label{eq:OriginalModel}
	\end{equation}
	where $\bm{\alpha}$ denotes the sparse code of $\mathbf{x}$ on the wavelet basis $\mathbf{A}$~($\mathbf{x}=\mathbf{A}\bm{\alpha}$), and $\lambda $ indicates the trade-off parameter. We use the $\ell_p$ regularization with $p \in(0,1)$ so that the problem turns out to be challenging non-convex sparse optimization. To tackle the challenge, our framework embraces imaging principles, deep inference, and sparsity priors, either of which was proved to be critical for effective reconstruction in previous studies. Accordingly, we design the fidelity, data-driven, and prior modules in every optimization iteration to exploit these three types of information. Further, the  framework embeds an optimality conditional module that automatically checks the output of each iteration and rejects the improper one leading to the undesired result. This module maintains the convergence, yielding a reliable solution to the optimization.	
	

	The top row of Fig.~\ref{fig:FlowChart} demonstrates the cascade of all ingredients at each iteration of the optimization for our framework, and we elaborate the four modules in the order shown as Fig.~\ref{fig:FlowChart}.
	
	\textbf{Fidelity module:} The fidelity term reveals the image formation process of MR imaging in~\eqref{eq:OriginalModel}, which is necessary for reconstruction algorithms. This module in our framework constitutes a mapping at the ($k$+1)th iteration from the previous output $\bm{\alpha}^{k}$~(initialized as the degraded input $\bm{\alpha}^{0}= \mathbf{A}^\mathrm{T}\mathbf{F}^\mathrm{T}\mathbf{P}^\mathrm{T}\mathbf{y}$) to an intermediate reconstructed image $\bm{u}^{k+1}$, characterizing the inverse imaging process as: 
	\begin{equation}
	\mathbf{u}^{k+1} = \mathcal{F}\left(\bm{\alpha}^{k};\theta_{\mathcal{F}}\right),
	\label{eq:fidelity}
	\end{equation}
	where $\theta_{\mathcal{F}}$ denotes the parameters for the mapping. The calculation of the mapping $\mathcal{F}$ and its parameters $\theta_{\mathcal{F}}$ will be detailed later in the reconstruction process.

	\textbf{Data-driven module:} Networks provide a flexible mechanism for characterizing underlying structures shared by numerous training pairs. Our framework consists of deep architectures encoding these image priors. We regard the artifacts introduced by the fidelity mapping $\mathcal{F}$ as unknown noise, and thus employ a residual denoising network with shortcuts (IRCNN)~\cite{zhang2017learning} as the data-driven module of our framework:  
	\begin{equation}
	\mathbf{v}^{k+1} = \mathcal{N}\left(\mathbf{u}^{k+1};\theta_{\mathcal{N}}^{k+1}\right),
	\end{equation}
	where $\theta_{\mathcal{N}}^{k+1}$ denotes the parameters of the residual network in the ($k$+1)th stage. We will give a detailed depiction of the choice for this CNN-based denoiser later. 
	
	\textbf{Optimal condition module:} One major issue arises with the  data-driven module whether the inference of the deep network $\mathcal{N}$ follows the direction converging to the optimum of the model energy. This module equipped with a checking mechanism resolves this issue, via automatically checking and rejecting the improper output of the data-driven module:
	%
	\begin{equation}
	\mathbf{w}^{k+1}= \mathcal{C}( \mathbf{v}^{k+1} ,\bm{\alpha}^{k}; \theta_{\mathcal{C}}), 
	\label{eq:w}
	\end{equation}
	where $\mathbf{v}^{k+1}$ and $\bm{\alpha}^{k}$ are the outputs of the data-driven module at this stage and of the previous iteration, respectively. $ \theta_{\mathcal{C}}$ denotes the parameters for the optimality condition. This module calculates an indicator using the first-order optimality condition that determines whether to accept the updating from deep architectures. 
	
	\textbf{Prior module:}
	The sparsity prior upon a specific basis is necessary to constrain the solution space in traditional CS-MRI. The prior naturally describes the intrinsic mutual characteristics of MR images so that reconstruction can benefit from incorporating this prior. Thus, we append a prior module to the condition module to further improve restoration quality:
	\begin{equation}
	\bm{\alpha}^{k+1} =\mathcal{P}\left( \mathbf{w}^{k+1};\theta_{\mathcal{P}}\right),
	\end{equation}
	where $\theta_{\mathcal{P}}$ denotes the parameters for the prior. Consequently, we are able to include domain knowledge for MR images, and recover more details.
	We apply an iterative process with all these four modules cascaded to solve the optimization~\eqref{eq:OriginalModel} as shown in the bottom row of Fig.~\ref{fig:FlowChart}. The converged optimum $\bm{\alpha}^{*}$, found by the process, gives the final reconstruction.

	
	\subsection{Reconstruction process}
	We elaborate the computation of each module in our deep optimization framework, forming an efficient and converged reconstruction process.
	
	\textbf{Closed solution for module $\mathcal{F}$:}	This module intends to give an intermediate solution to the fidelity term in~\eqref{eq:OriginalModel}, $\frac{1}{2}\|\mathbf{PFA}\bm{\alpha}\!-\!\mathbf{y}\|_2^2$, without the complex sparsity constraint. Instead, we constrain the solution close to the reconstruction at the previous iteration $\bm{\alpha}^{k}$ using the $l_2$ norm:
	\begin{equation}
	\mathbf{u}^{k+1}  =
	\arg\min\limits_{\mathbf{u}}
	\frac{1}{2}\|\mathbf{PFAu}\!-\!\mathbf{y}\|_2^2\!+
	\!\frac{\rho}{2}\|\mathbf{u}\!-\!\bm{\alpha}^k\|_2^2,
	\label{eq:ModelWithProxTerm}
	\end{equation}
	where $\rho$ balances between the fidelity term and $\bm{\alpha}^{k}$. This constraint guarantees the recovery from the fidelity term not deviating from the overall optimization direction. Moreover, the continuity of~\eqref{eq:ModelWithProxTerm} renders a closed solution:
	\begin{equation}
	\begin{aligned}
	\mathbf{u}^{k+1\!}\! = \!\mathcal{F}\left(\bm{\alpha}^{k};\!\rho\right)\!=\!
	\mathbf{A}^\mathrm{T}\mathbf{F}^\mathrm{T}\!\left(\mathbf{P}^\mathrm{T}\mathbf{P}\!+\!\rho \mathbf{I}\right)^{\!-1\!}
	\left(\mathbf{P}^\mathrm{T}\mathbf{y}\!+\!\rho \mathbf{FA}\bm{\alpha}^k\right).
	\end{aligned}
	\label{eq:closed}
	\end{equation}
	
	\textbf{Residual learning for module $\mathcal{N}$:}
	Inspired by the discovery that artifacts from under-sampled data have a topologically simpler structure than original images~\cite{lee2017deep}, we choose a deep architecture with shortcuts to learn structural information underlying pre-reconstructed images $\mathbf{u}^{k+1}$ with noise/artifacts. The deep architecture consists of seven blocks. The first block is composed of a dilated convolution layer cascaded with a rectified linear unit (ReLU) layer. Each of the five middle blocks consists of three layers,~\emph{i.e.}, a dilated convolution, a batch normalization and a ReLU. The last one is a single convolution layer.  Dilated convolutions with a larger receptive field are adopted in the network to differentiate anatomical structures from artifacts, resulting in better denoising performance.  
	
	Such a deep approach accommodates various image structures reflected by training examples without the need of explicit models. More importantly, a well trained deep network can provide rather fast forward inference needing no gradient calculation in the image domain at each iteration. Specifically, we formulate this forward process as:
	\begin{equation}
	\mathbf{v}^{k\!+\!1} = \mathcal{N}\left(\mathbf{u}^{k\!+\!1};\vartheta^{k\!+\!1}\right) = \mathbf{A}^\mathtt{T}(\emph{ResNet}(\mathbf{A}\mathbf{u}^{k\!+\!1};\vartheta^{k\!+\!1})),
	\end{equation}
	where $\emph{ResNet}$ denotes the inference via the trained deep network with shortcuts. The output $\mathbf{u}^{k+1}$ of the fidelity term is a set of coefficients upon the wavelet domain rather than an image, and thus we have to convert $\mathbf{u}^{k+1}$ into the image domain by  $\mathbf{A}\bm{u}^{k+1}$ as the input fed to the network. Subsequently, we convert the network output back into coefficients by applying the inverse transformation of $\mathbf{A}$.
	
	\textbf{Optimal condition for module $\mathcal{C}$:}
	To maintain the theoretical convergence of the proposed framework, we design a checking mechanism to guarantee iterations always converging towards the optimal solution of the original energy~\eqref{eq:OriginalModel}. First, we introduce a proximal gradient of a momentum term to connect the output of the data-driven module with the first-order optimal condition of the minimization energy. We define the momentum proximal gradient\footnote{$\mathtt{prox}_{\eta \lambda\|\cdot\|_{p}}(\mathbf{v}) = \arg\min_{\mathbf{x}} \lambda  \|\mathbf{x}\|_{p} +\frac{1}{2} \|\mathbf{x} - \mathbf{v}\|^2$.} as
	\begin{equation}
	\bm{\beta}^{k+1\!}\!\in\!\mathtt{prox}_{\!\eta_1\lambda\|\!\cdot\!\|_p}
	\!\left(\!\mathbf{v}^{k+1\!}\!-
	\!\eta_1\!\left(\!\nabla f\!\left(\!\mathbf{v}^{k+1\!}\right)\!+
	\!\rho\!\left(\!\mathbf{v}^{k+1\!}\!-\!\bm{\alpha}^{k\!}\right)\!\right)\!\right),
	\label{eq:checkingprox}
	\end{equation}
	where $\eta_1$ is the step-size and $f$ denotes the fidelity term in~\eqref{eq:OriginalModel}. Then, we establish a feedback mechanism by considering the first-order optimal condition of~\eqref{eq:checkingprox} as 
	\begin{equation}
	\|\mathbf{v}^{k+1}-\bm{\beta}^{k+1}\| \leq \varepsilon^{k}\|\bm{\alpha}^{k}-\bm{\beta}^{k+1}\|.
	\label{eq:error}
	\end{equation}
	Here, $\varepsilon^{k}$ is a positive constant revealing the tolerance scale of the distance between the current solution $\bm{\beta}^{k+1}$ and previous updating $\bm{\alpha}^{k} $ at the $k$-th stage. The previous output $\bm{\alpha}^{k}$ is also re-considered when~\eqref{eq:error} is not satisfied. Finally, we summarize the checking module as
	\begin{equation}
	\begin{array}{l}
	\mathbf{w}^{k+1}= \mathcal{C}( \mathbf{v}^{k+1} ,\bm{\alpha}^{k}; \eta_1) 
	= \left\{
	\begin{array}{ll}
	\bm{\beta}^{k+1}    & \ \text{if}~\eqref{eq:error}~\text{is satisfied}\\
	\bm{\alpha}^{k}     & \ \text{otherwise}.\\
	\end{array} \right.
	\end{array}
	\label{eq:w}
	\end{equation}
	In this manner, the iterative process bypasses the improper output of deep networks, and directs to the desired solution.
	
	\textbf{Proximal gradient for module $\mathcal{P}$:} The sparsity prior $\lambda\|\mathbf{v}\|_p$ in~\eqref{eq:OriginalModel} acts as the last module for each iteration. Considering the non-convexity of $\ell_p$-norm regularization, we solve it via a step of proximal gradient as follows:
	\begin{equation}
	\!\bm{\alpha}^{k+1}\! =\!\mathcal{P}\!\left( \mathbf{w}^{k+1};\!\eta_2\right)\! \in \! \mathtt{prox}_{\eta_2 \lambda\|\cdot\|_p}\!\left(\mathbf{w}^{k+1}\!-\!\eta_2\nabla f\!\left(\mathbf{w}^{k+1}\!\right)\!\right),
	\end{equation}
	where $\eta_2$ is the step size. This operation enforces the prior term of the reconstruction model, and thus preserves more details avoiding over-smoothing reconstruction.
	
	Algorithm~\ref{alg1} lists the iterative reconstruction process where each iteration consists of cascaded computations for the four principled modules. This process converges to the solution to~\eqref{eq:OriginalModel}, and we will give the theoretical analysis below.
	
	\begin{algorithm}[!t]
		\caption{Reconstruction procedure} 
		\label{alg1}
		\begin{algorithmic}[1]
			\REQUIRE $ \mathbf{x^0, P, F, y,}$ and necessary parameters.
			\ENSURE Reconstructed MR image $ \mathbf{x}$. 
			\WHILE{ not converged}
			\STATE $\mathbf{u}^{k+1}  =  \mathcal{F}\left(\bm{\alpha}^{k};\rho\right), $
			\STATE $\mathbf{v}^{k+1}  = \mathcal{N}\left(\mathbf{u}^{k+1};\vartheta^{k+1}\right), $
			\STATE $\mathbf{w}^{k+1} =\mathcal{C}\left( \mathbf{v}^{k+1} , \bm{\alpha}^{k}; \eta_1\right), $
			\STATE $\bm{\alpha}^{k+1} = \mathcal{P}\left( \mathbf{w}^{k+1};\eta_2\right),$

			\ENDWHILE
			\STATE $\mathbf{x} = \mathbf{A}\bm{\alpha}^{*}. $
		\end{algorithmic}
	\end{algorithm}

	%
	
	\subsection{Theoretical investigations}
	Existing deep optimization strategies highly rely on the distribution of data, and discard the convergence guarantee in iterations~\cite{sun2016deep,diamond2017unrolled}. Our scheme not only includes energy optimization with imaging principles and prior knowledge, but also embraces a learnable deep architecture for fast inference. Furthermore, we design an effective mechanism to determine whether the output of networks at current iteration follows a descent direction towards the energy optimum. This subsection theoretically analyzes the convergence behavior of our method. 
	
	To simplify the following derivations, we first rewrite the function in~\eqref{eq:OriginalModel} as
	$$
	\begin{array}{l}
	\Phi (\bm{\alpha}) =\frac{1}{2}\| \mathbf{PF\mathbf{A}\bm{\alpha}-y}\|_2^2+\lambda \|\bm{\alpha}\|_p.
	\end{array}
	$$
	Furthermore, we also give the following properties about $\Phi$ that are helpful for the convergence analysis\footnote{We place the details of these properties, and all the proofs of the following propositions and theorem in supplemental materials due to page limit.}:\\
	a) $\frac{1}{2}\| \mathbf{PF\mathbf{A}\bm{\alpha}-y}\|_2^2$ is proper and Lipschitz smooth;\\	
	b) $\lambda \|\bm{\alpha}\|_p$ is proper and lower semi-continuous;\\
	c) $\Phi (\bm{\alpha})$ satisfies the K{\L} property and is coercive.\\
	
	Then two important propositions and one theorem are given to characterize the convergence properties of our method.
	
	\begin{proposition}\label{prop:c-error}
		Let $ \left\{\bm{\alpha}^k\right\}_{k\in\mathbb{N}} $ and  $\left\{\bm{\beta}^k\right\}_{k\in\mathbb{N}} $ be the sequences generated by Alg.~\ref{alg1}. Supposing that the error condition 
		$\|\mathbf{v}^{k+1}-\bm{\alpha}^{k}\| \leq \varepsilon^{k}\|\bm{\beta}^{k+1}-\bm{\alpha}^{k}\|$ 
		in our module $\mathcal{C}$ is satisfied, there exists a sequence $\{C^{k}\}_{k\in\mathbb{N}}$ such that 
		\begin{equation}	
		\Phi(\bm{\beta}^{k+1}) \leq \Phi(\bm{\alpha}^k)-C^{k}\|\bm{\beta}^{k+1}-\bm{\alpha}^k\|^2,
		\end{equation}
		where $C^{k} = {1}/{2\eta_{1}} - {L_f}/{2} -  (L_f  + |\rho-{1}/{\eta_{1}}|)\epsilon^{k} >0$ and $L_f$ is the Lipschitz coefficient of $\nabla f$ .
		\label{eq:ineq_fun_pgmomentum}
	\end{proposition}
	
	\begin{proposition}\label{prop:pg}
		If $\eta_2 < 1/L_f$, let $ \left\{\bm{\alpha}^k\right\}_{k\in\mathbb{N}} $ and $\left\{\mathbf{w}^k\right\}_{k\in\mathbb{N}}$ be the sequences generated by a proximal operator, then we have
		\begin{equation}
		\Phi(\bm{\alpha}^{k+1}) \leq \Phi(\mathbf{w}^{k+1})-({1}/({2\eta_2}) - {L_f}/{2})\|\bm{\alpha}^{k+1}-\mathbf{w}^{k+1}\|^2.\label{eq:ineq_fun_pg}
		\end{equation}	
	\end{proposition}
	\begin{remark}
		The inequalities in Propositions \ref{prop:c-error} and \ref{prop:pg} provide a descent sequence $\Phi(\bm{\alpha}^{k})$ by illustrating the relationship of $\Phi(\bm{\alpha}^{k})$ and $\Phi(\mathbf{w}^k)$ as
		$$
		-\infty <\Phi(\bm{\alpha}^{k+1}) \leq \Phi(\mathbf{w}^k) \leq \Phi(\bm{\alpha}^k) \leq \Phi(\bm{\alpha}^0).
		$$
		Thus, the operator $\mathcal{C}$ is a key criterion to check the output of the learnable deep module whether to propagate along a descent direction toward the optimum. Moreover, it also ingeniously builds a bridge to connect the adjacent iteration $\Phi(\bm{\alpha}^{k})$ and $\Phi(\bm{\alpha}^{k+1})$.
	\end{remark}
	
	\begin{theorem}\label{theorem:convergence}
		Suppose $ \left\{\bm{\alpha}^k\right\}_{k\in\mathbb{N}} $ be a sequence generated by our algorithm. The following properties hold.
		\begin{itemize}
			\item The square summable of sequence $\left\{\bm{\alpha}^{k+1}-\mathbf{w}^{k+1} \right\}_{k\in\mathbb{N}}$ is bounded, i.e., 
			$
			\sum_{k=1}^{\infty}\|\bm{\alpha}^{k+1}-\mathbf{w}^{k+1}\|^2 < \infty.
			$
			\item The sequence $ \left\{\bm{\alpha}^k\right\}_{k\in\mathbb{N}} $ converges to a critical point $\bm{\alpha}^{*}$ of $\Phi$.
		\end{itemize}
	\end{theorem}
	
	\begin{remark}
		The second property in Theorem~\ref{theorem:convergence} implies that $\left\{\bm{\alpha}^k\right\}_{k\in\mathbb{N}}$ is a Cauchy sequence so that it globally converges to the critical point of $\Phi$. 
	\end{remark}

	\section{Extensions to practical scenarios}
	Our framework can adapt to practical scenarios having similar models with~\eqref{eq:OriginalModel}. This section gives two examples,~\emph{i.e.}, parallel imaging and reconstruction with Rician noise.
	\subsection{Extension 1: Parallel imaging based CS-MRI}
	Reconstruction of sparse multi-coil data can be formed as:
	\begin{equation}
	\bm{\alpha}\in
	\arg\min\limits_{\bm{\alpha}}\left\{\sum_{l=1}^{L}
	\ \frac{1}{2}\|\mathbf{PFS}_l\mathbf{A}\bm{\alpha}\!-\!\mathbf{y}_l\|_2^2\!+
	\!\lambda\|\bm{\alpha}\|_p\right\},
	\label{eq:ParallelImagingModel}
	\end{equation}
	where $L$ is the number of receivers used for multi-coil data acquisition. $\mathbf{y}_l \in \mathbb{C}^M$ denotes the under-sampled MR data recorded by the $l^{th}$ receiver with the corresponding sensitivity map $\mathbf{S}_l \in \mathbb{C}^{N \times N}$ which can be estimated via commonly used self-calibration techniques~\cite{uecker2014espirit,el2018self}. 
	
	This CS-PI model \eqref{eq:ParallelImagingModel} shares a similar form with~\eqref{eq:OriginalModel} so that we are able to readily apply the proposed framework for efficient and reliable reconstruction.
	The fidelity term in~\eqref{eq:ParallelImagingModel} derives the closed-form solution for the module $\mathcal{F}$ as:
	\begin{equation}
	\begin{aligned}
	\mathbf{u}^{k+1\!}\! = \!\mathcal{F}\left(\bm{\alpha}^{k};\!\rho\right)\!=\!
	\sum_{l=1}^{L}\frac{\mathbf{A}^\mathrm{T}{\mathbf{S}_l}^\mathrm{T}
		\mathbf{F}^\mathrm{T}\left(\mathbf{P}^\mathrm{T}\mathbf{y}_l\!+\rho\mathbf{F}\mathbf{S}_l\mathbf{A}\bm{\alpha}^k\right)}
	{\mathbf{P}^\mathrm{T}\mathbf{P}\!+\!\rho \mathbf{I}}.
	\end{aligned}
	\label{eq:PI-closed}
	\end{equation}
	
	Accordingly, the derivative of the fidelity term, $\nabla f(\bm \alpha) $, in the modules $\mathcal{C}$ and $\mathcal{P}$ changes into:
	\begin{equation}
	\begin{aligned}
	\nabla f(\bm \alpha) = \sum_{l=1}^{L}\mathbf{A}^\mathrm{T}{\mathbf{S}_l}^\mathrm{T}\mathbf{F}^\mathrm{T}
		\mathbf{P}^\mathrm{T}\left(\mathbf{P}\mathbf{F}\mathbf{S}_l\mathbf{A}\bm\alpha-\mathbf{y}_l\right).
	\end{aligned}
	\label{eq:PI-closed}
	\end{equation}
	
	We can find final reconstruction by running Alg.~\ref{alg1} with other computations unchanged.
	
	\subsection{Extension 2: CS-MRI with Rician noise}
	In order to enable the proposed framework to tackle the reconstruction with non-addition Rician noise, we re-formulate the CS-MRI model in~\eqref{eq:OriginalModel} as:
	\begin{equation}
	\begin{array}{l}
	\min\limits_{\mathbf{x}}\ 
	\frac{1}{2}\| \mathbf{PF}\mathbf{z}-
	\mathbf{y}\|_2^2+\lambda_1\|\mathbf{A}_1^{\top}\mathbf{x}\|_p +
	\lambda_2\|\mathbf{A}_2^{\top}\mathbf{z}\|_p  \\
	s.t.\ \ \ \ \mathbf{z}=\sqrt{(\mathbf{x}+\mathbf{n_1})^2+\mathbf{n}_2^2},
	\end{array}
	\label{eq:rician1}
	\end{equation}
	where $\mathbf{x}$, $\mathbf{y}$ and $\mathbf{n_1,n_2}\sim\mathcal{N}(0,\sigma^2)$ denote the fully sampled clear MR image, observed $k$-space data and independent Gaussian noise in real and imaginary components, respectively. The matrices $\mathbf{A}_1^{\top}$ and $\mathbf{A}_2^{\top}$ are the inverse of wavelet transforms $\mathbf{A}_1$ and $\mathbf{A}_2$, respectively.
	
	We then split ~\eqref{eq:rician1} into two subproblems:
	\begin{subequations}
		\begin{numcases}{}
		\begin{array}{l}
		\!\!\!\!\mathbf{z}^{k+1} \! =\!\arg\min\limits_{\mathbf{z}}
		\frac{1}{2}\| \mathbf{PFz}-\mathbf{y}\|_2^2 +
		\lambda_1\|\mathbf{A}_1^{\top}\mathbf{z}\|_p\\
		\qquad \quad +\frac{\rho_1}{2}\|\mathbf{z}-  \mathbf{z}^{k} \|_2^2,
		\end{array}\label{eq:rician_sub_z}\\
		\begin{array}{l}
		\!\!\!\!\mathbf{x}^{k+1}\!  =\!\arg\min\limits_{\mathbf{x}}
		\frac{1}{2}\|\sqrt{(\mathbf{x}\!\!+\!\!\mathbf{n}_1)^2 +
			\mathbf{n}_2^2}-\mathbf{z}^{k+1}\|_2^2\\
		\qquad \quad+\lambda_2\|\mathbf{A}_2^{\top}\mathbf{x}\|_p,
		\end{array}\label{eq:rician_sub_x}
		\end{numcases}
	\end{subequations}
	where $\rho_1$ is a penalty parameter. Therefore, we can alternatively apply the iterative process in Alg.~\ref{alg1} to solve $\mathbf{z}^{k+1}$ and $\mathbf{x}^{k+1}$, finally converging to the optimal estimate $\mathbf{x}^{*}$. 
	
	The subproblem~\eqref{eq:rician_sub_z} aims to reconstruct a fully sampled \emph{noisy} MR image from the $k$-space observation $\mathbf{y}$. Therefore, Alg.\ref{alg1} is applicable to the subproblem with all modules unchanged but taking the two square terms as the fidelity module:
	\begin{equation}
	\mathbf{u_z}^{k\!+\!1}\!  =\!
	\arg\min\limits_{\mathbf{u_z}}
	\!\frac{1}{2}\!\|\mathbf{PFu_z}\!-\!\mathbf{y}\|_2^2\!+
	\!\frac{\rho_1}{2}\!\|\mathbf{u_z}\!-\!\mathbf{z}^k\|_2^2 \!+
	\!\frac{\rho_2}{2}\!\|\mathbf{u_z}\!-\!\mathbf{x}^k\|_2^2,
	\label{eq:ModelWithProxTermRicien}
	\end{equation}
	where $\mathbf{z}^k$ and $\mathbf{x}^k$ represent the output at previous iteration of the subproblems~\eqref{eq:rician_sub_z} and~\eqref{eq:rician_sub_x}, respectively. $\rho_1$ and $\rho_2$ balances among the fidelity term, $\mathbf{z}^k$ and $\mathbf{x}^k$. Also, the continuity of~\eqref{eq:ModelWithProxTermRicien} renders a closed solution:
	\begin{equation}
	\begin{aligned}
	\!\mathbf{u_z}^{\!k\!+\!1}\! &= \!\mathcal{F}\left(\mathbf{z}^{k}, \mathbf{x}^{k};\!\rho\right)\!\\
	&=\!
	\mathbf{F}^\mathrm{\!T}\!\left(\!\mathbf{P}^\mathrm{\!T}\mathbf{P}\!+\!\rho_1\mathbf{\!I}+\!\rho_2\mathbf{\!I}\right)^{\!-1\!}
	\!\left(\!\mathbf{P}^\mathrm{\!T}\mathbf{y}\!+\!\rho_1 \mathbf{F}\mathbf{z}^{k}\!+\!\rho_2 \mathbf{F} \mathbf{x}^{k}\!\right).
	\end{aligned}
	\label{eq:closedRicien}
	\end{equation}
	Accordingly, we derive the derivative of the fidelity term $\nabla f(\mathbf{z}) $ in both modules $\mathcal{C}$ and $\mathcal{P}$ as:
	\begin{equation}
	\begin{aligned}
	\nabla f(\mathbf{z}) = \mathbf{F}^\mathrm{T}
	\left(\mathbf{P}^\mathrm{T}\mathbf{P}\mathbf{F}\mathbf{z}-\mathbf{P}^\mathrm{T}\mathbf{y}\right)+
	\rho_1\left(\mathbf{z}-\mathbf{z}^{k}\right) .
	\end{aligned}
	\label{eq:derivationRician}
	\end{equation}
	
	The second subproblem~\eqref{eq:rician_sub_x} targets at noise removal by writing the fidelity module as:
	\begin{equation}
	\mathbf{u_x}^{k\!+\!1}\!  =\!
	\arg\min\limits_{\mathbf{u_x}}
	\!\frac{1}{2}\|\!\sqrt{\!\left(\mathbf{u_x}\!\!+\!\!\mathbf{n}_1\right)^2 \!+\!
		\mathbf{n}_2^2}\!-\!\mathbf{z}^{k\!+\!1}\|_2^2\\ \!.
	\label{eq:ModelWithProxTermRicienx}
	\end{equation}
	Unfortunately, we can hardly express a closed-form solution $\mathbf{u_x}$ to the module $\mathcal{F}$, and thus resort to an efficient learnable strategy with two-stage IRCNNs, trained using pairs of noise-free and Gaussian noisy MR images, in order to tackle non-additive Rician noise. The first stage predicts $(\mathbf{u_x} + \mathbf{n}_{1})^{2}$ from $(\mathbf{z}^{k+1})^{2}$, which removes additive noise $\mathbf{n}_{2}^{2}$ by noticing $(\mathbf{z}^{k+1})^{2}=(\mathbf{u_x} + \mathbf{n}_{1})^{2} + \mathbf{n}_{2}^{2}$. The second stage feeds the square root of the output of the first stage,~\emph{i.e.}, $\sqrt{(\mathbf{u_x} + \mathbf{n}_{1})^{2} }$, to the trained IRCNN, inferring $\mathbf{u_x}^{k+1}$ that eliminates additive noise $\mathbf{n}_{1}$. Given $\mathbf{u_x}^{k+1}$, we obtain one step iteration $\mathbf{x}^{k+1}$ by subsequently calculating the other three modules, $\mathcal{N}$, $\mathcal{C}$ and $\mathcal{P}$ in Alg.~\ref{alg1}. The gradient $\nabla f(\mathbf{x})$ in $\mathcal{C}$ and $\mathcal{P}$ changes into:
	\begin{equation}
	\begin{aligned}
	\nabla f(\mathbf{x}) \!=\! (\mathbf{x}\! + \!\mathbf{n_1})(1\! -\! (\sqrt{(\mathbf{x} \!+\!\mathbf{n_1})^2\!+\!\mathbf{n_2}^2})^{\!-1\!}\mathbf{z}^{k+1})),
	\end{aligned}
	\label{eq:derivationRicianx}
	\end{equation}
	where $\mathbf{n_1}$ and $\mathbf{n_2}$ are handily available as byproducts of the two-stage IRCNN inference.

	\begin{figure}[!tbp]
		\begin{center}
			\begin{tabular}{c@{\extracolsep{-1em}}c@{\extracolsep{0.2em}}c}
				&\includegraphics[width=.24\textwidth]{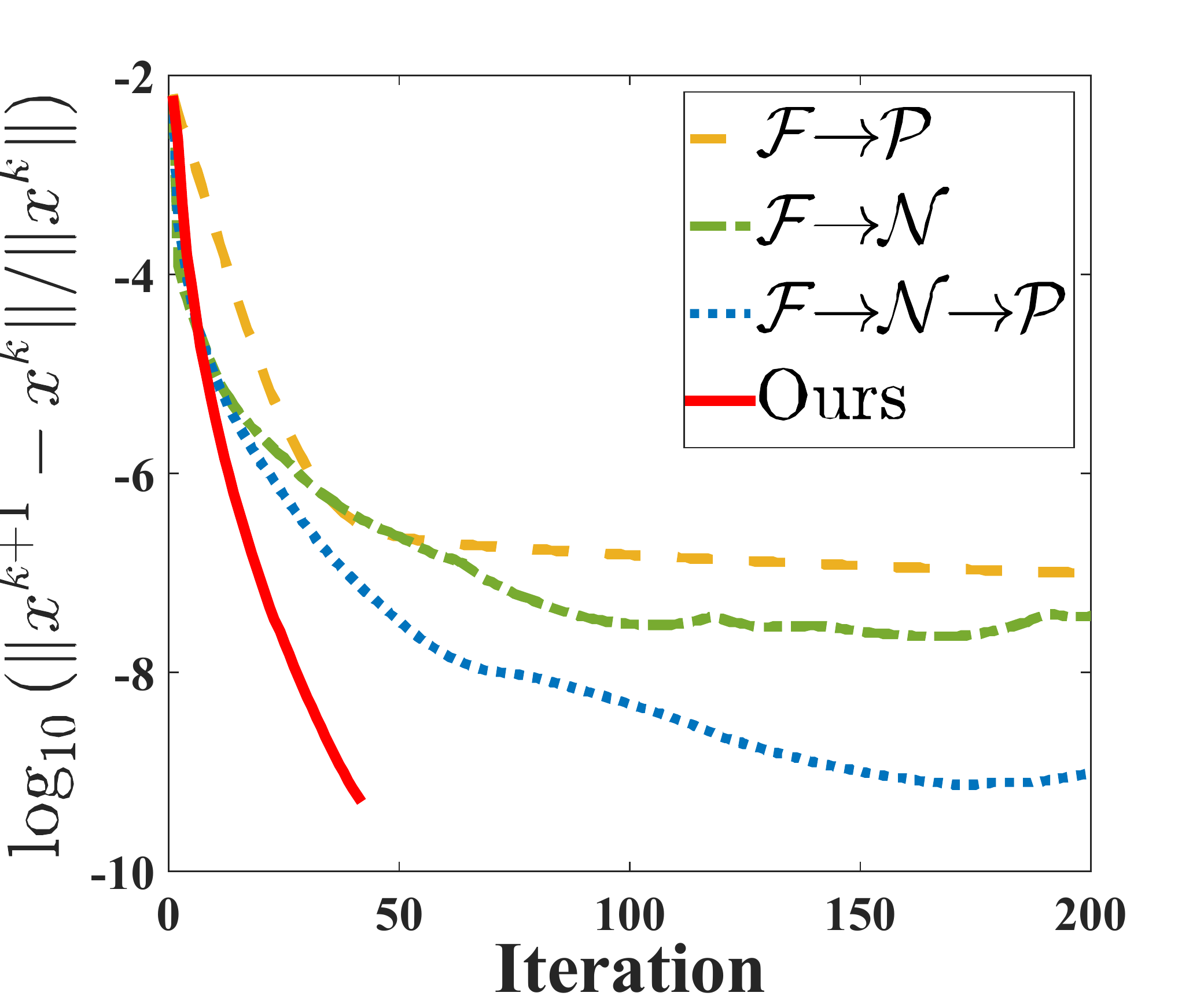}
				&\includegraphics[width=.237\textwidth]{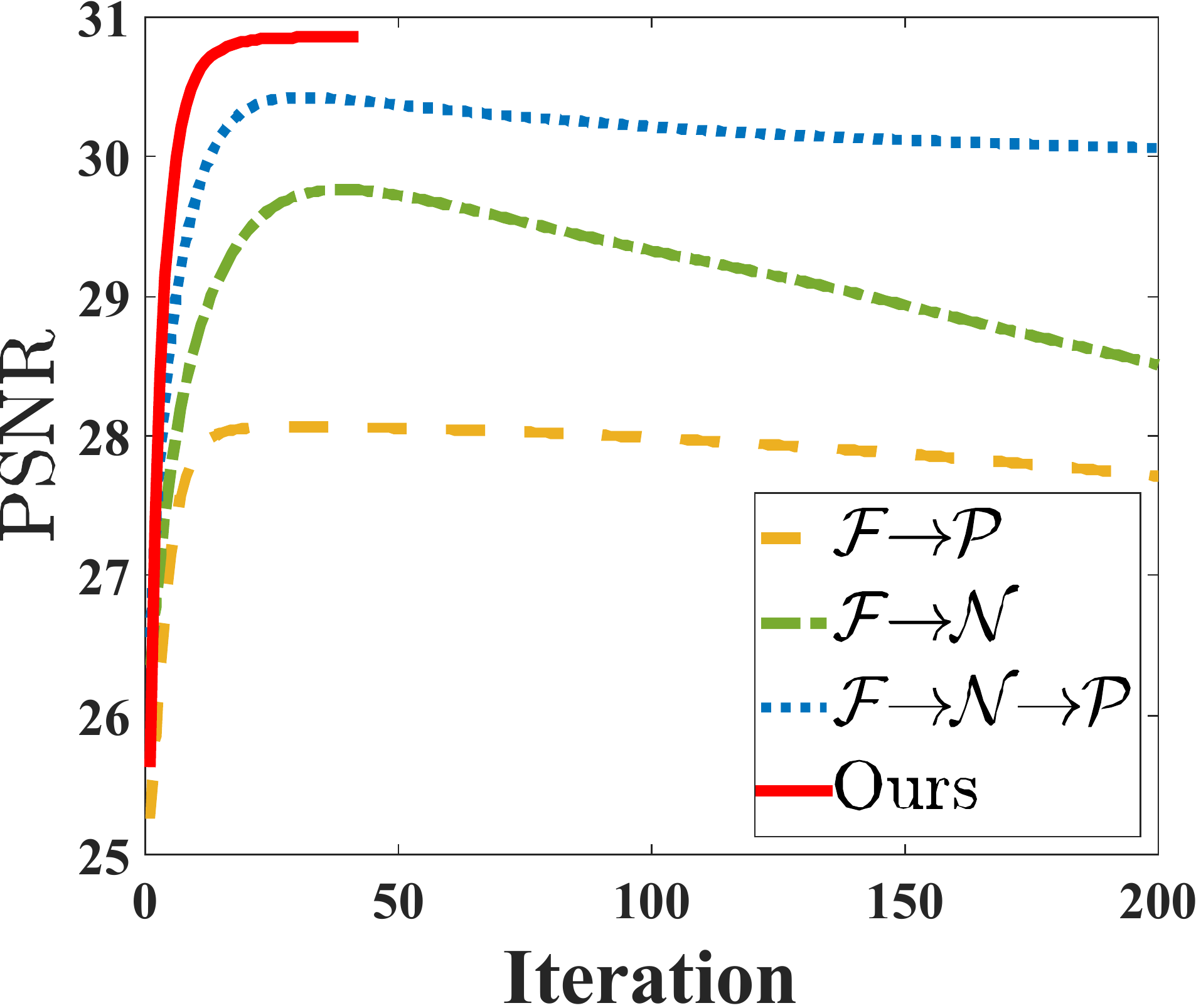}\\
			\end{tabular}
		\end{center}
		\caption{Loss (left) and PSNR (right) evolution during iterations of four module combinations.}
		\label{iteration}
	\end{figure}
	\begin{figure}[!tbp]
		\begin{center}
			\begin{tabular}{c@{\extracolsep{-0.5em}}c@{\extracolsep{0.3em}}c@{\extracolsep{0.3em}}c@{\extracolsep{0.3em}}c@{\extracolsep{0.0em}}c}
				&\includegraphics[width=.11\textwidth]{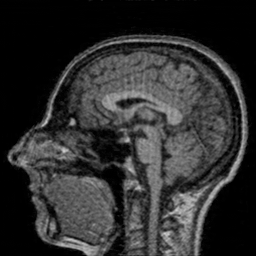}
				&\includegraphics[width=.11\textwidth]{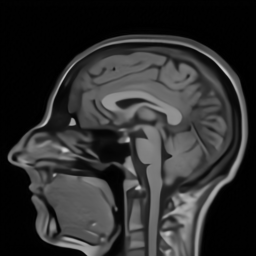}
				&\includegraphics[width=.11\textwidth]{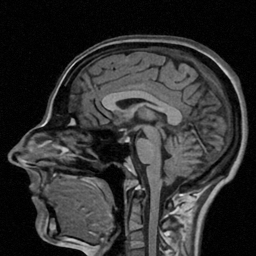}
				&\includegraphics[width=.11\textwidth]{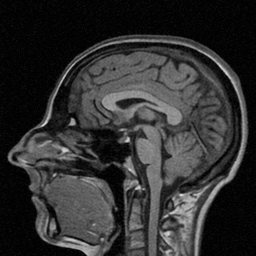}\\
				& $\mathcal{F}\!\!\to\!\!\mathcal{P}$&$\mathcal{F}\!\!\to \!\!\mathcal{N}$& $\mathcal{F}\!\!\to\!\!\mathcal{N}\!\!\to\!\!\mathcal{P} $ & Ours\\
			\end{tabular}
		\end{center}
		\caption{Images reconstructed with four module combinations.}
		\label{component}
	\end{figure}

	\begin{table*}[t]
		\centering
		\caption{Quantitative comparison on different sampling patterns at 20\% sampling ratio.}
		\vspace{-0.5em}
		\begin{tabular}{|p{2.4cm}|p{0.7cm}p{0.9cm}|p{0.7cm}p{0.9cm}|p{0.7cm}p{0.9cm}|p{0.7cm}p{0.9cm}|p{0.7cm}p{0.9cm}|p{0.7cm}p{0.9cm}|} 
			\hline
			Data &	 \multicolumn{6}{c|}{T$_1$-weighted MR Data} & 	\multicolumn{6}{c|}{T$_2$-weighted MR Data}	\\ \hline
			Mask 	 	& \multicolumn{2}{c|}{Cartesian} 	& \multicolumn{2}{c|}{Radial}	& \multicolumn{2}{c|}{Gaussian}	& \multicolumn{2}{c|}{Cartesian}	& \multicolumn{2}{c|}{Radial}	& \multicolumn{2}{c|}{Gaussian} 	\\ \hline	
			Evaluation	& PSNR &	RLNE  & PSNR  &RLNE   & PSNR  &RLNE 	& PSNR &RLNE 	& PSNR &RLNE 	& PSNR &RLNE 	\\ \hline
			
			ZeroFilling~\cite{bernstein2001effect} 	&22.16 &	0.2874 &	25.19 &	0.2026 &	25.08 &	0.2051& 	23.80 &	0.3724 & 	25.69 &	0.2636 &	26.48 &	0.2737 \\
			TV~\cite{lustig2007sparse}				&23.85 &	0.2366 &	27.76 &	0.1514 &	29.76 &	0.1204& 	26.76 &	0.2650 &	33.27 &	0.1258 &	35.85 &	0.0938 \\
			SIDWT~\cite{baraniuk2007compressive}	&23.43 &	0.2479 &	28.82 &	0.1449 &	29.64 &	0.1220& 	25.96 &	0.2905 &	33.51 &	0.1173 &	36.14 &	0.0905 \\
			PBDW~\cite{qu2012undersled}			&25.96 &	0.1857 &	30.86 &	0.1145 &	32.38 &	0.0890& 	29.18 &	0.2008 &	35.33 & 0.0951 &	38.54 &	0.0685 \\
			PANO~\cite{qu2014magnetic}				&27.88 &	0.1492 & 	30.51 &	0.1104 &	33.94 &	0.0746& 	31.74 &	0.1499 &	36.09 &	0.0912 &	40.43 &	0.0557 \\
			FDLCP~\cite{zhan2016fast}				&26.47 &	0.1757 &	30.75 &	0.1075 & 	26.47 &	0.1757& 	31.74 &	0.1500 &	37.33 &	0.0788 &	31.74 &	0.1500 \\
			BM3D-MRI~\cite{eksioglu2016decoupled}	&26.20 &	0.1802 &	31.08 &	0.1035 & 	35.41 &	0.0630& 	29.23 &	0.1999 &	37.55 &	0.0770 &	42.67 &	0.0428 \\
			DAMP~\cite{eksioglu2018denoising}		&25.90 &	0.1866 &	30.31 &	0.1128 & 	35.30 &	0.0638&		27.98 &	0.2306 &	35.98 &	0.0919 &	42.30 &	0.0447 \\
			ADMM-Net~\cite{sun2016deep}				&25.14 &	0.2041 &	29.42 &	0.1254 &	33.22 &	0.0807& 	28.62 &	0.2141 &	36.30 &	0.0888 &	39.19 &	0.0636 \\
			DAGAN~\cite{yang2017dagan}		&24.57 &	0.1707 &	26.79 &	0.1324 & 	27.79 &	0.1177&		25.89 &	0.1196 &	29.99 &	0.0748 &	30.40 &	0.0714 \\
			RefineGAN~\cite{quan2018compressed}	&25.27 &	0.2057 &	28.20 &	0.1462 &	26.72 &	0.1713& 	27.58 &	0.2506 &	33.54 & 0.1240 &	33.81 &	0.1226 \\
			Ours									&\textbf{28.22} &	\textbf{0.1433} &	\textbf{32.01} & \textbf{0.0928} &	\textbf{36.17} &	\textbf{0.0576}& 	\textbf{33.57} &	
										\textbf{0.1211} &	\textbf{38.47} &	\textbf{0.0691} & 	\textbf{42.77} &	\textbf{0.0422} \\
			\hline
		\end{tabular}
		\label{tab:patterns}
	\end{table*}

	\begin{figure}[!htbp] 
		\begin{center}
			\begin{tabular}{c@{\extracolsep{-1em}}c@{\extracolsep{0.8em}}c@{\extracolsep{2em}}c}
				&\includegraphics[width=.245\textwidth]{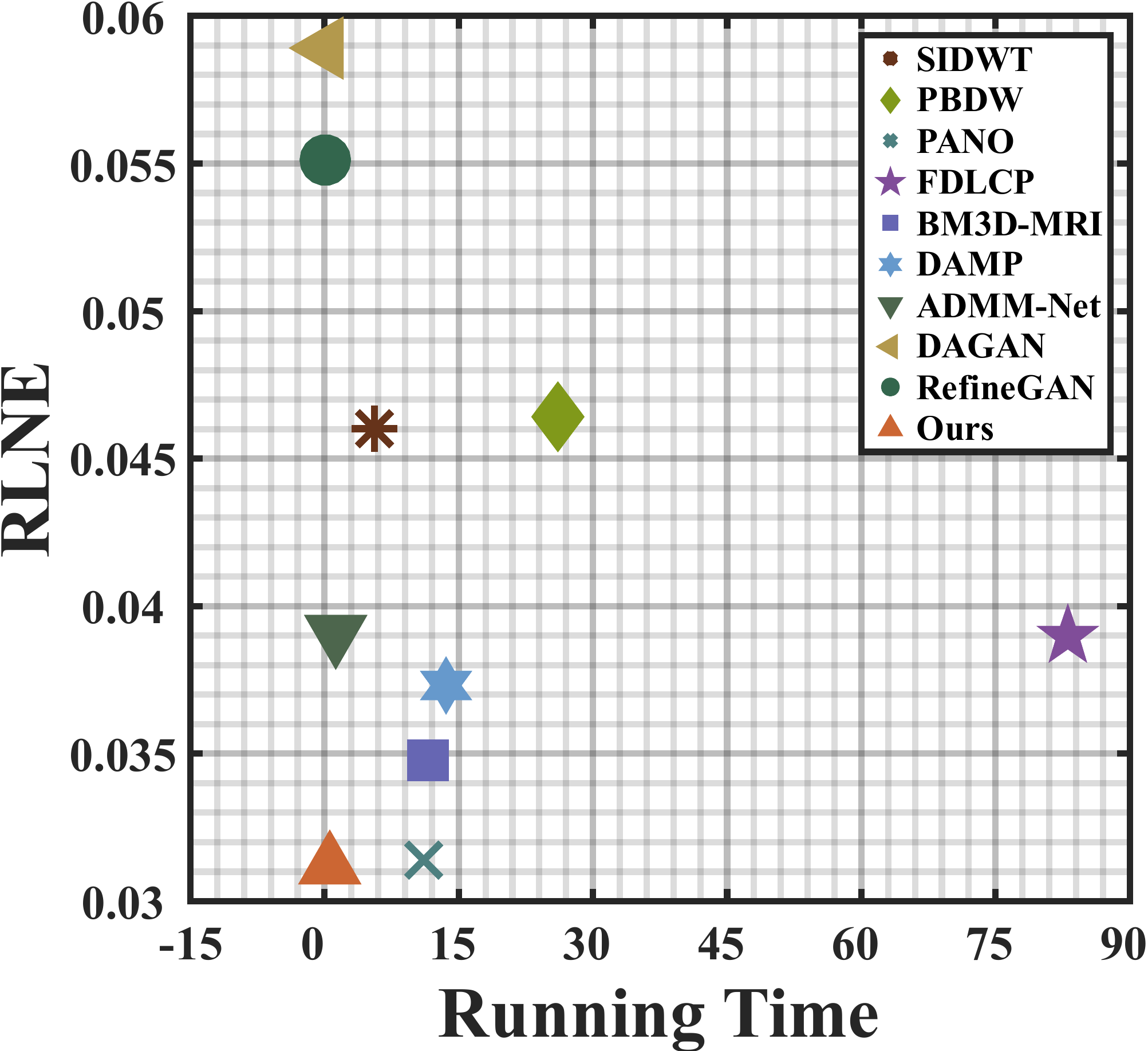}	
				&\includegraphics[width=.2325\textwidth]{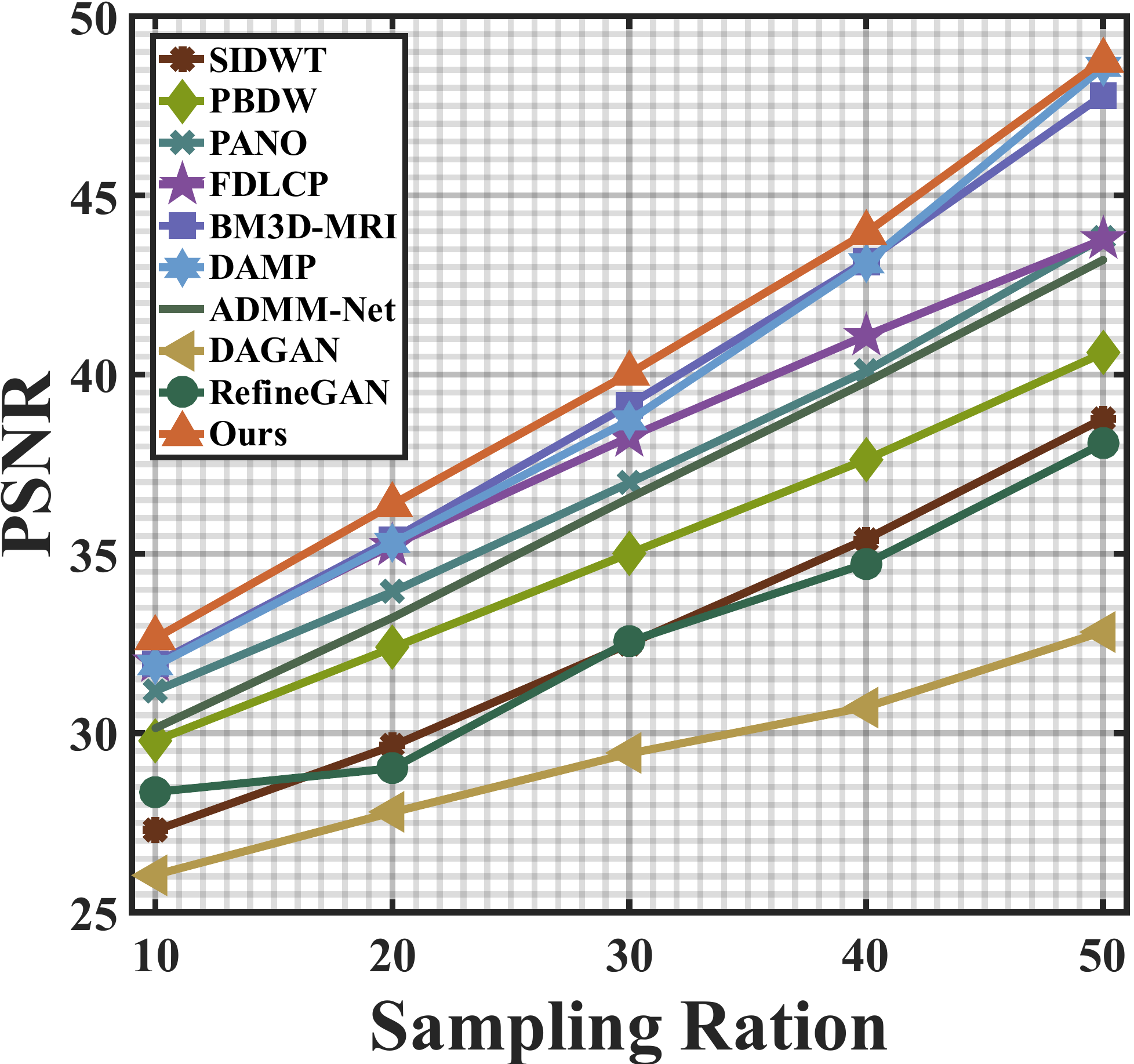}	\\	
			\end{tabular}
		\end{center}
		\caption{RLNE VS Running time (left) and PSNR VS sample ratios (right) reflect efficiency and robustness to sampling variations, respectively. We omit Zerofilling and TV whose results fall beyond the range of these two plots.}
		\label{fig:TimeRatio}
	\end{figure}
	
	\begin{table*}
		\centering
		\renewcommand\tabcolsep{10.5pt} 
		\caption{Comparisons of time consumption on T$_2$-weighted Data using Radial pattern at 50\% sampling ratio.}
		\begin{tabular}{|c|c|c|c|c|c|c|c|c|c|}
			\hline
			Evaluation 	&ZeroFilling& SIDWT 	& PBDW 		& PANO 		& FDLCP 	& BM3D-MRI 	& DMAP	& ADMM-Net 	& Ours \\
			\hline	
			PSNR		&22.85	&26.66&	26.55&	26.53	&25.65&	24.88&	26.23&	24.90&	\textbf{26.73}\\
			SSIM		&0.5353	&0.5863	&0.5553&	0.5804&	0.5293&	0.5674&	0.5282&	0.5431&	\textbf{0.5949}\\
			RLNE		&0.3062	&0.2339	&0.2359	&0.2289	&0.2653	&0.2729&	0.2529&	0.3085&\textbf{0.2186}\\						
			\hline	
		\end{tabular}
		\label{tab:RealData}
	\end{table*}

	\begin{figure*}[!htbp]
		\begin{center}
			\begin{tabular}{l@{\extracolsep{-0.5em}}c@{\extracolsep{0.2em}}
					c@{\extracolsep{0.2em}}c@{\extracolsep{0.2em}}c@{\extracolsep{0.2em}}c@{\extracolsep{0.2em}}c@{\extracolsep{0.2em}}c}
				&\includegraphics[width=.13595\textwidth]{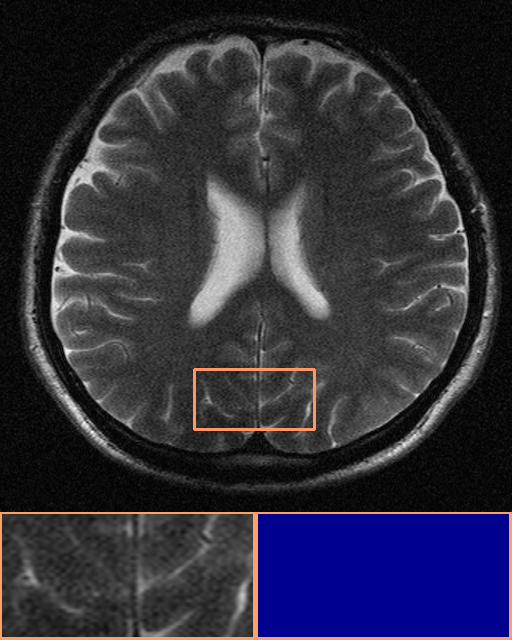}
				&\includegraphics[width=.13595\textwidth]{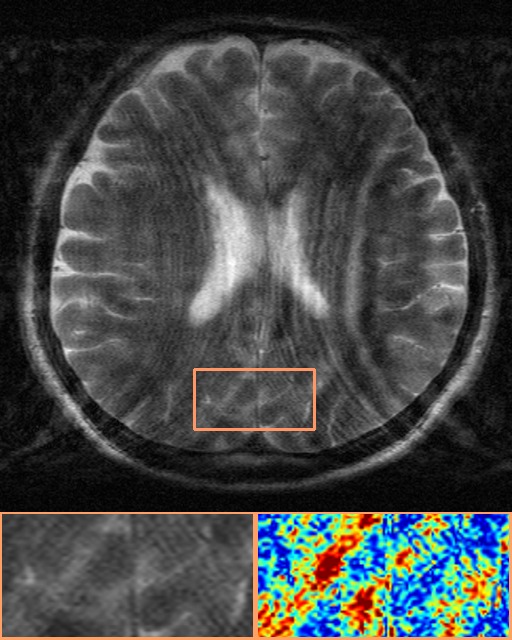}
				&\includegraphics[width=.13595\textwidth]{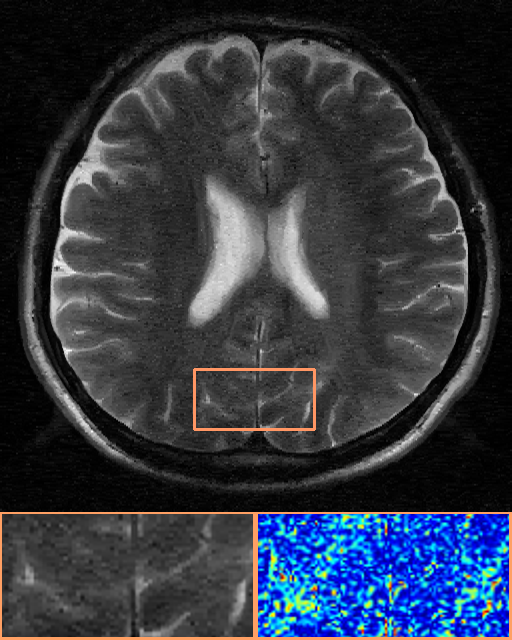}
				&\includegraphics[width=.13595\textwidth]{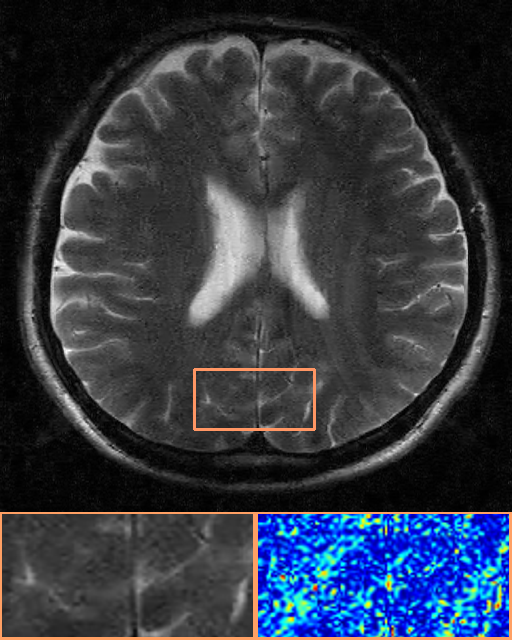}
				&\includegraphics[width=.13595\textwidth]{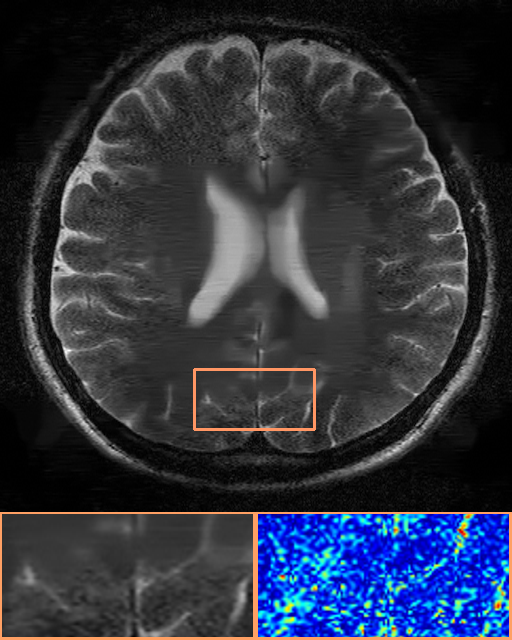}
				&\includegraphics[width=.13595\textwidth]{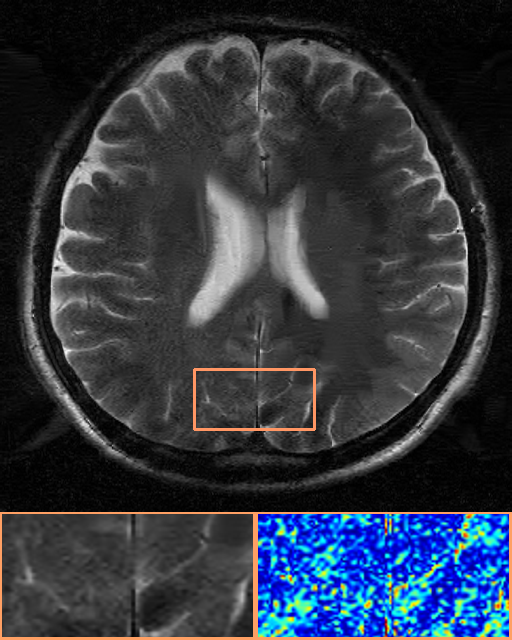}
				&\includegraphics[width=.13595\textwidth]{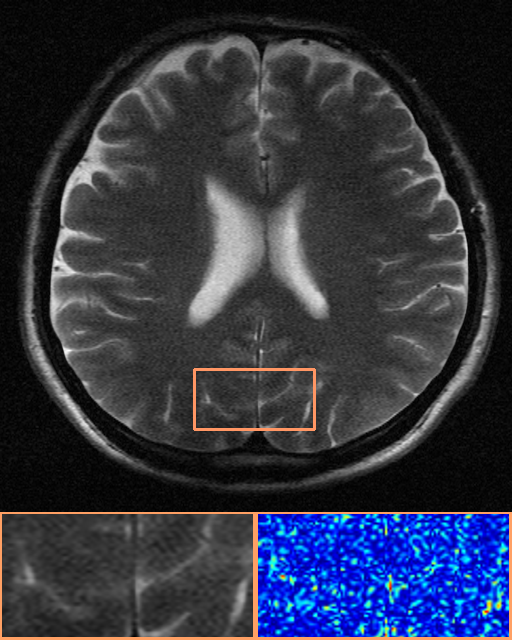}\\
				&Ground Truth & ZF (23.94) & PBDW (27.21) & PANO (27.66)& FDLCP (26.75)	 & DAMP (27.20)  & Ours (\textbf{28.33})\\			
			\end{tabular}
		\end{center}
		\caption{Images, zoomed-in details, and error heat maps (blue-lower and red-higher) reconstructed from raw $k$-space data at the acceleration factor of 2.0. Resultant PSNR values are parenthesized after corresponding algorithms.}
		\label{fig:RealData}
	\end{figure*}
	
	\begin{table}
		\centering
		\renewcommand\tabcolsep{6.3pt} 
		\caption{Comparisons on parallel imaging}
		\begin{tabular}{|c|c|c|c|c|c|} 
			\hline
			Evaluation&SENSE&ESPIRiT&ESPIRiT-L$_1$ &GRAPPA & Ours\\
			\hline		
			PSNR& 27.82 & 29.14  & 30.13 &  30.96  &  31.02  \\
			RLNE& 0.31 & 0.27  & 0.25 &  00.23  &  0.23  \\
			Time& 8.5570 & 9.8833  & 234.982 &  116.407  & 9.1078     \\
			\hline
		\end{tabular}
		\label{tab:parallelImaging}
	\end{table}
	
	\begin{figure*}[!htbp]
		\begin{center}
			\begin{minipage}{1\textwidth}
				\begin{tabular}{l@{\extracolsep{-0.5em}}c@{\extracolsep{0.2em}}c@{\extracolsep{0.2em}}
						c@{\extracolsep{0.2em}}c@{\extracolsep{0.2em}}c@{\extracolsep{0.2em}}c@{\extracolsep{0.2em}}c@{\extracolsep{0.2em}}c}
					&\includegraphics[width=.119\textwidth]{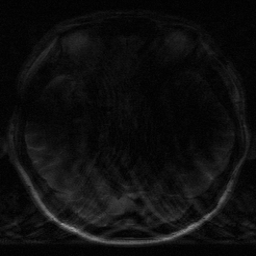}
					&\includegraphics[width=.119\textwidth]{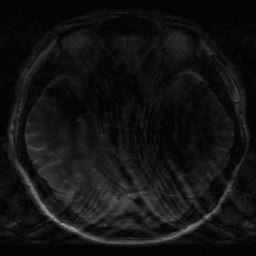}
					&\includegraphics[width=.119\textwidth]{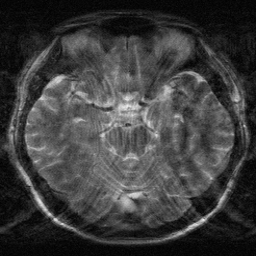}
					&\includegraphics[width=.119\textwidth]{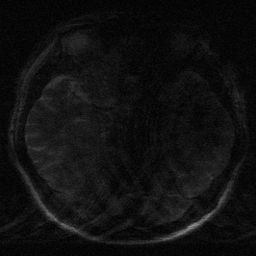}
					&\includegraphics[width=.119\textwidth]{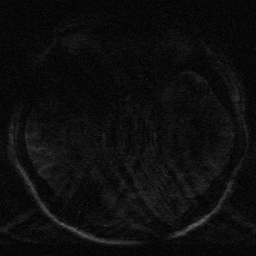}
					&\includegraphics[width=.119\textwidth]{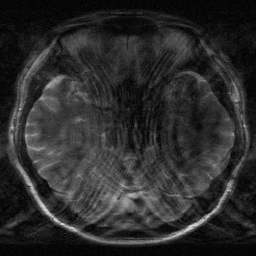}
					&\includegraphics[width=.119\textwidth]{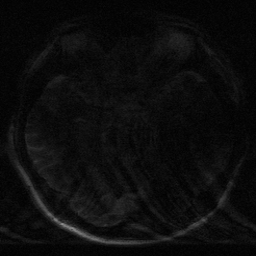}
					&\includegraphics[width=.119\textwidth]{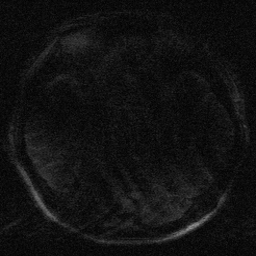}	\\				
				\end{tabular}
			\end{minipage}
			
			\begin{minipage}{1\textwidth}
				\begin{tabular}{l@{\extracolsep{-0.5em}}c@{\extracolsep{0.2em}}c@{\extracolsep{0.2em}}
						c@{\extracolsep{0.2em}}c@{\extracolsep{0.2em}}c@{\extracolsep{0.2em}}c}
					&\includegraphics[width=.16\textwidth]{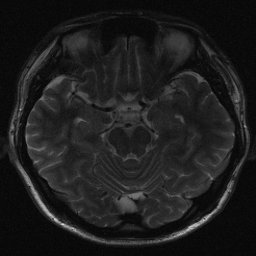}
					&\includegraphics[width=.16\textwidth]{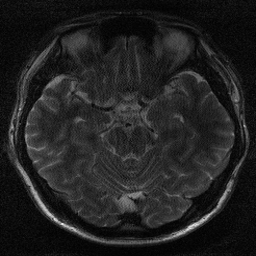}
					&\includegraphics[width=.16\textwidth]{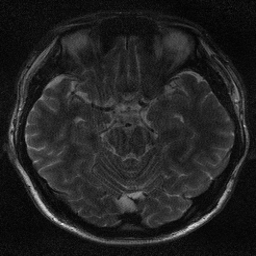}c
					&\includegraphics[width=.16\textwidth]{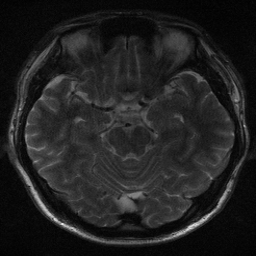}
					&\includegraphics[width=.16\textwidth]{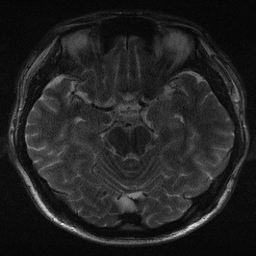}
					&\includegraphics[width=.16\textwidth]{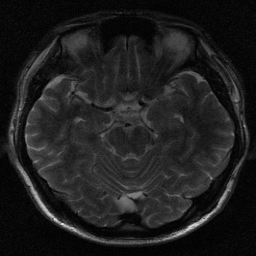}\\
					&GT&SENSE (27.43) & ESPIRiT(29.59)  &ESPIRiT-L$_1$(30.59)  & GRAPPA (30.43) & Ours (\textbf{31.961})\\			
				\end{tabular}
			\end{minipage}
		\end{center}
		\caption{Images for parallel imaging from multi-channel data collected by 8 coils with the acceleration factor of 2.0. PNSR values are given in parentheses after corresponding algorithms.}
		\label{fig:parallelImaging}
	\end{figure*}

\begin{figure}[!htbp]
	\begin{center}
		\begin{tabular}{l@{\extracolsep{-0.65em}}c@{\extracolsep{0.1em}}c@{\extracolsep{0.1em}}c@{\extracolsep{0.1em}}c@{\extracolsep{0.1em}}c@{\extracolsep{0.1em}}c}
			&\includegraphics[height=.058\textheight]{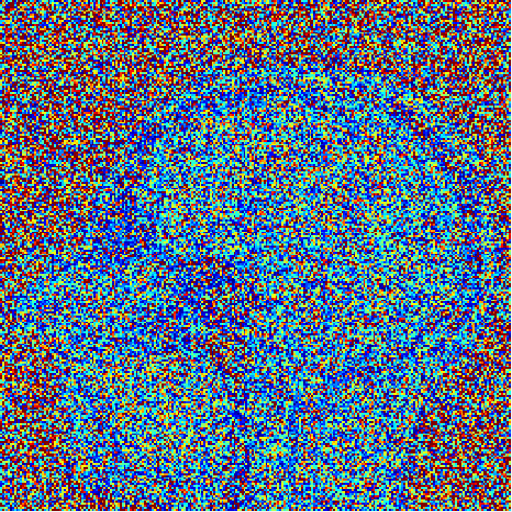}
			&\includegraphics[height=.058\textheight]{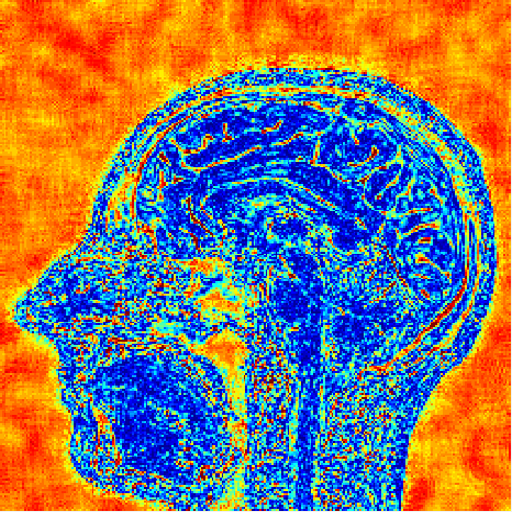}
			&\includegraphics[height=.058\textheight]{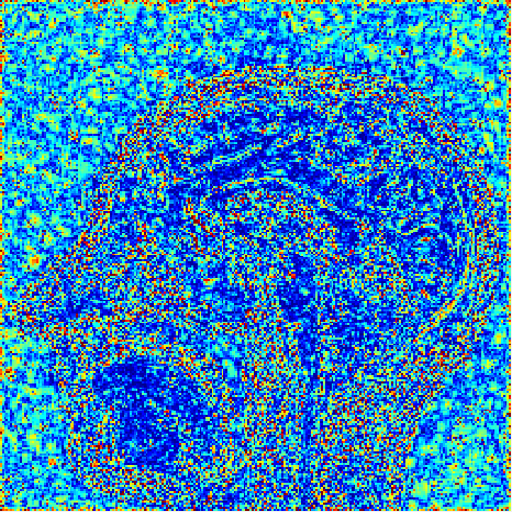}
			&\includegraphics[height=.058\textheight]{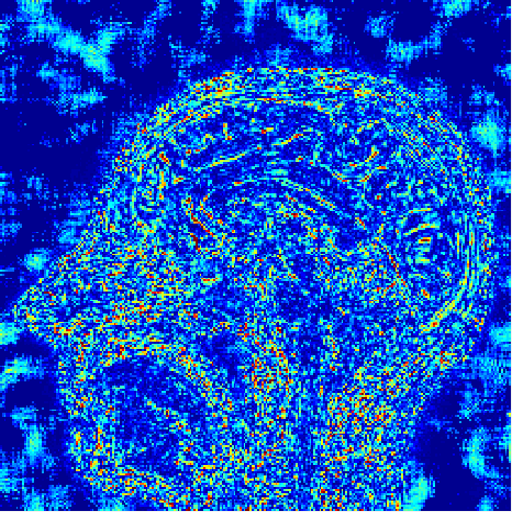}
			&\includegraphics[height=.058\textheight]{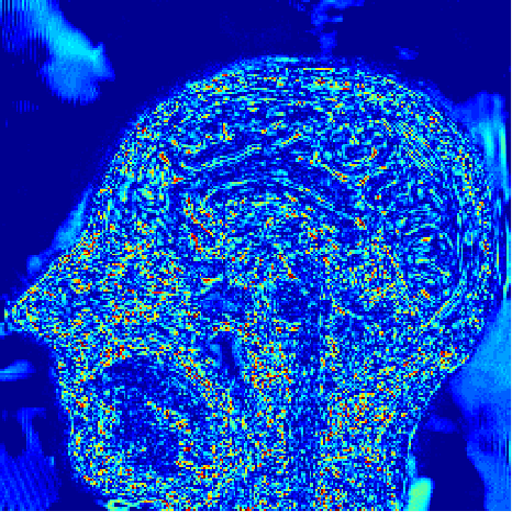}
			&\includegraphics[height=.0582\textheight]{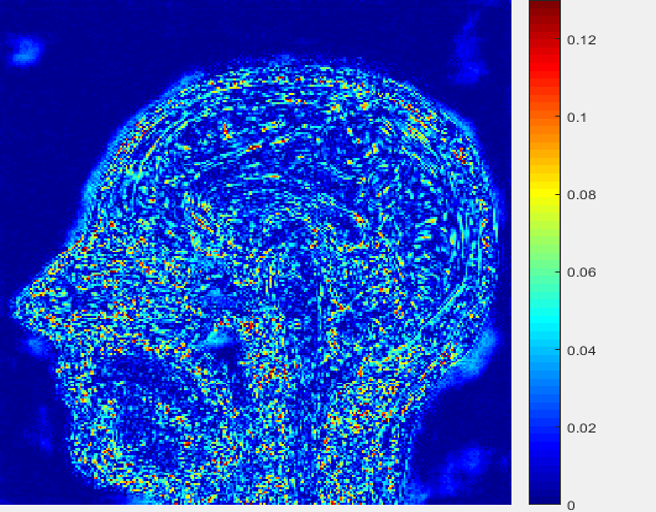}\\
			& Noisy & LMMSE & NLM & UNLM & RiceVST & Ours\\
			& (20.77) & (22.82) & (25.76) & (29.23) & (29.93) & \textbf{(30.84)}\\
		\end{tabular}
	\end{center}
	\caption{Heat maps (Blue-lower, Red-higher) of Rician denoisers with the noise level of 20. PSNR values are given below corresponding algorithms.}
	\label{fig:NoisyError}
\end{figure}

\begin{figure}[!htbp]
	\begin{center}
		\begin{tabular}{l@{\extracolsep{-1em}}c@{\extracolsep{0.6em}}c@{\extracolsep{0em}}c}
			&\includegraphics[height=.15\textheight]{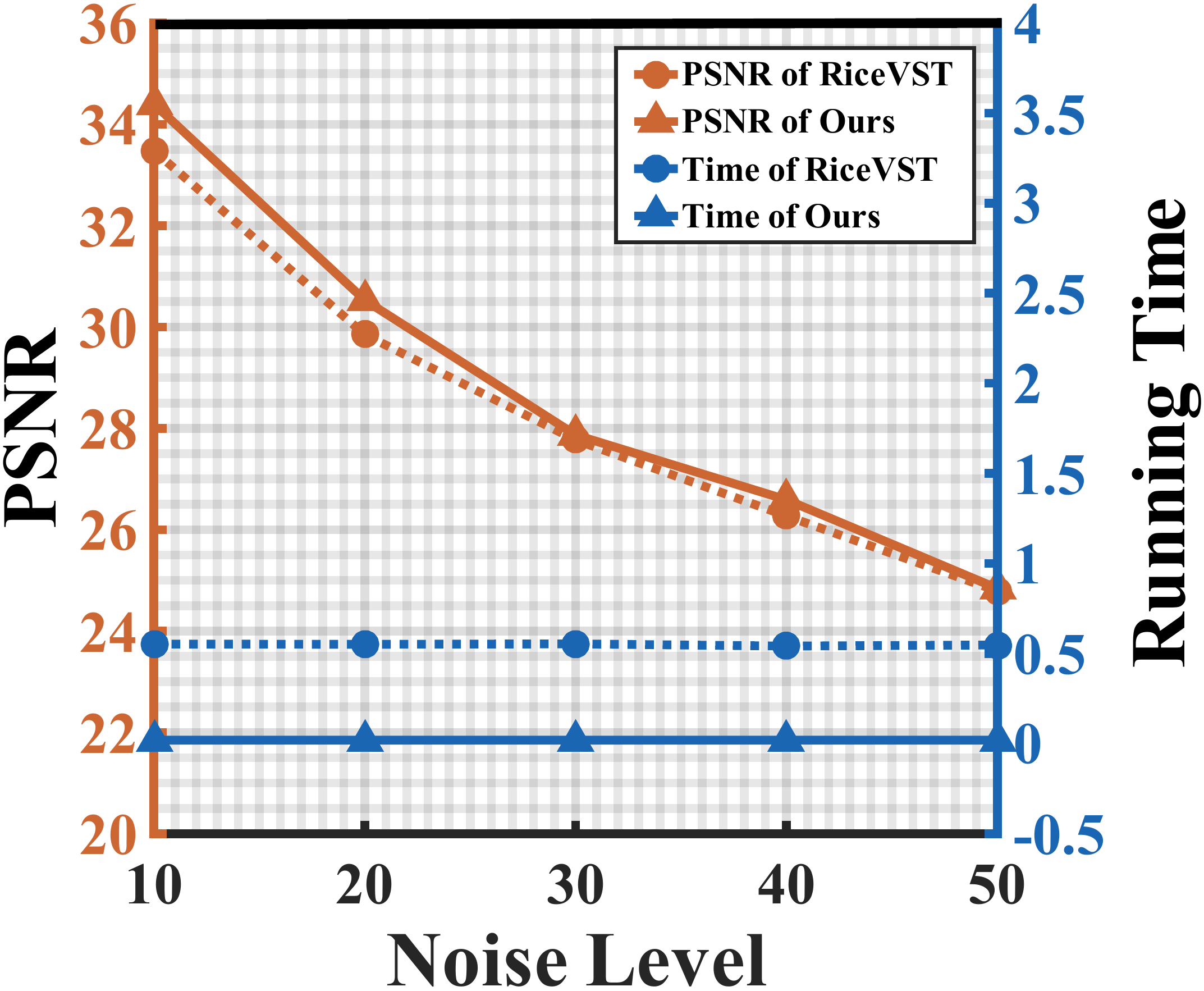}
			&\includegraphics[height=.15\textheight]{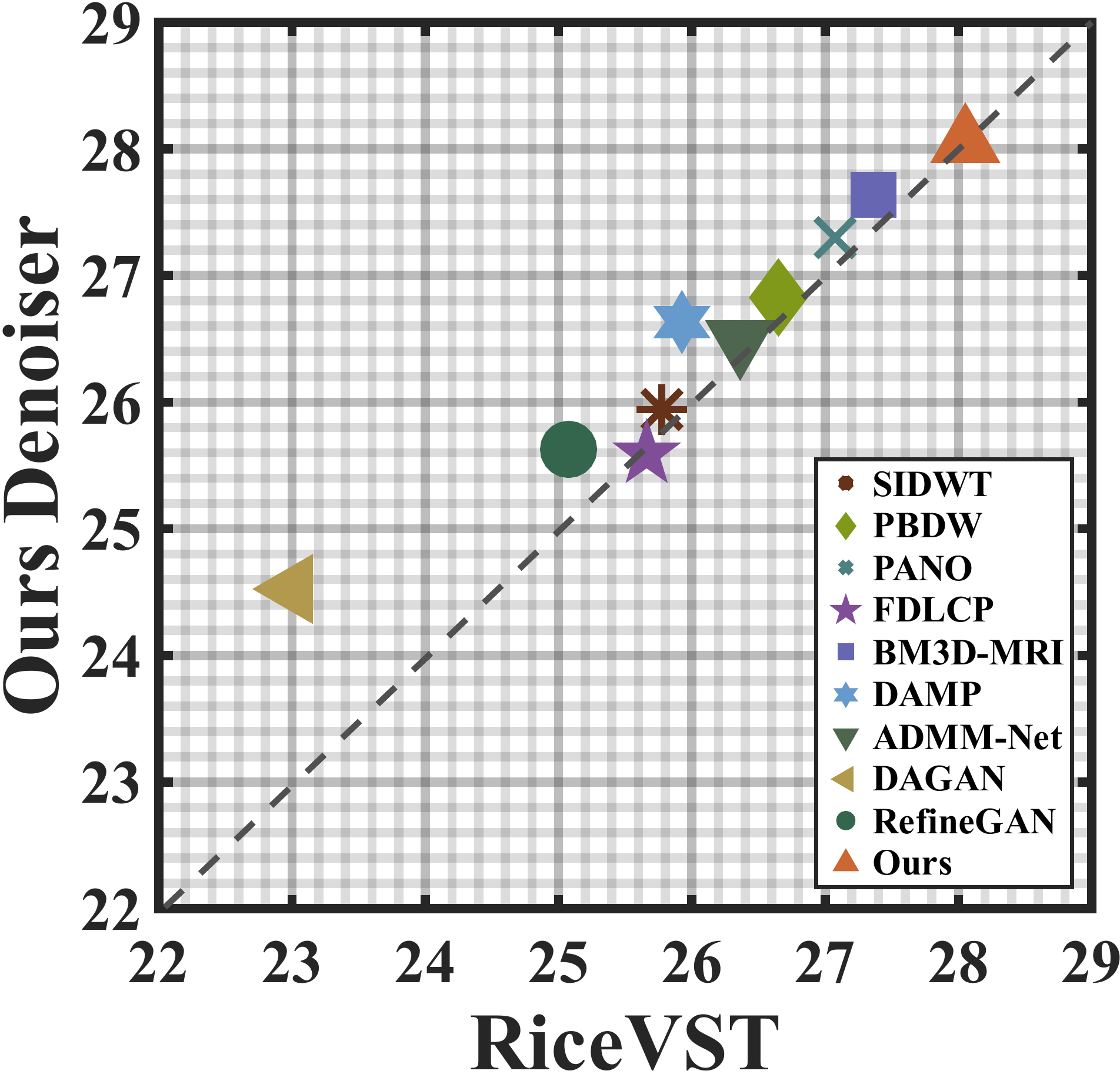}\\
		\end{tabular}
	\end{center}
	\caption{Left: PSNR and running time at different noise levels. Right: PSNR values by reconstruction algorithms with RiceVST ($x$-axis) and the two-stage deep denoiser ($y$-axis) cascaded as post-processing at the Rician noise level 20. Ours embeds the deep denoiser in every iteration.}
	\label{fig:NoisySigma}
\end{figure}

	\begin{figure*}[!htbp]
	\begin{center}
		\begin{tabular}{l@{\extracolsep{-0.2em}}c@{\extracolsep{0.25em}}c@{\extracolsep{0.25em}}
				c@{\extracolsep{0.25em}}c@{\extracolsep{0.25em}}c@{\extracolsep{0.25em}}c@{\extracolsep{0.25em}}c}
			&\includegraphics[width=.16\textwidth]{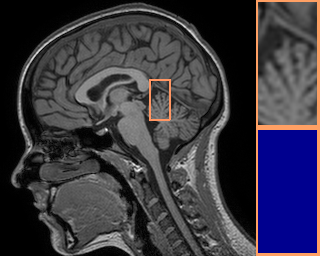}
			&\includegraphics[width=.16\textwidth]{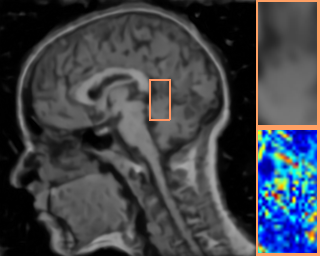}
			&\includegraphics[width=.16\textwidth]{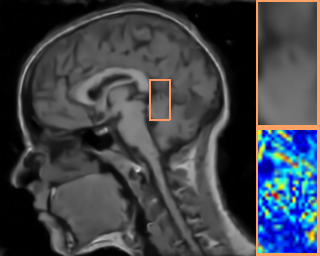}
			&\includegraphics[width=.16\textwidth]{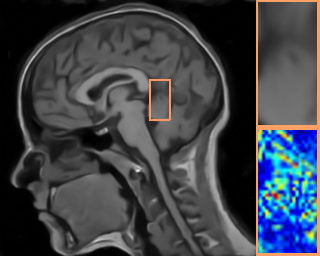}
			&\includegraphics[width=.16\textwidth]{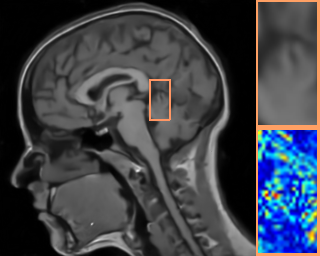}
			&\includegraphics[width=.16\textwidth]{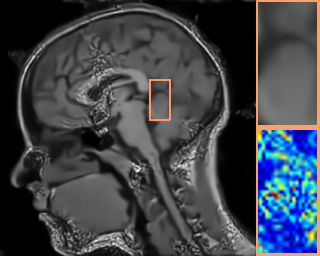}\\
			&Ground Truth & ZeroFilling (24.22)  & SIDWT (25.50) & PBDW (26.64) & PANO (27.03)& FDLCP (25.00)\\
			&\includegraphics[width=.16\textwidth]{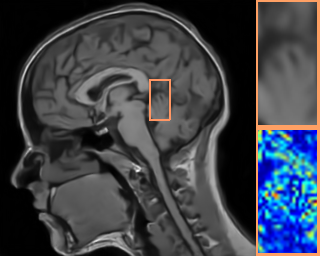}
			&\includegraphics[width=.16\textwidth]{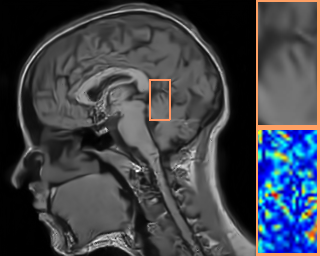}
			&\includegraphics[width=.16\textwidth]{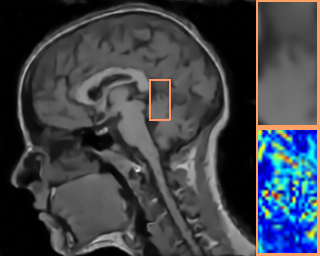}
			&\includegraphics[width=.16\textwidth]{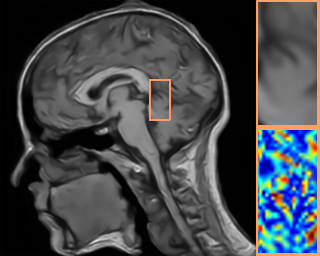}
			&\includegraphics[width=.16\textwidth]{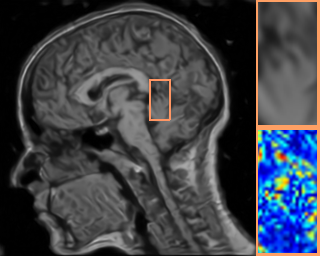}
			&\includegraphics[width=.16\textwidth]{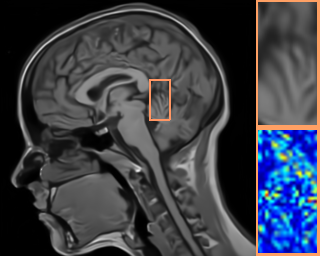}\\
			& BM3D-MRI (27.36)& DAMP (26.27) & ADMM-Net(26.09)&DAGAN(24.19)&RefineGAN(25.12)& Ours\textbf{(27.88)}\\
		\end{tabular}
	\end{center}
	\caption{Images, zoomed-in details, and heat maps reconstructed by traditional CS-MRI methods with the two-stage deep denoiser as post-processing. Ours combines reconstruction with the denoiser in one framewok.The parenthesized value gives the resultant PSNR.}
	\label{fig:noisyDetail}
\end{figure*}

\section{Experiments}
	This section first explores the impact of each module in our framework on accuracy and convergence behavior. 
	Next, we compare our algorithm with the state-of-the-art CS-MRI methods to demonstrate its superiority on accuracy, efficiency, and robustness. We also demonstrate on raw $k$-space data provided by an MR manufacturer Prodiva 1.5T scanner (Philips Healthcare, Best, Netherlands) besides benchmarks.
	Further, we conduct experiments on multi-coils data and Rician noisy data. 
	All experiments were executed on a PC with Intel(R) CPU @3.0GHz 256GB RAM and a NVIDIA TITAN Xp\footnote{Source codes and testing data are available at https://github.com/dlut-dimt/TGDF-TMI}.

\subsection{Ablation analysis}
	First we evaluate four combinations of modules in our framework to figure out their roles. The first case $\mathcal{F}\!\!\to\!\!\mathcal{P}$ performs reconstruction with the constraint of sparsity prior, and the second one $\mathcal{F}\!\!\to\!\!\mathcal{N}$ combines the fidelity and data-driven modules in order to demonstrate the performance of deep networks. We use both sparsity and deep priors without the optimal condition module as the third case $\mathcal{F}\!\!\to\!\!\mathcal{N}\!\!\to\!\!\mathcal{P}$. The entire paradigm with all modules is the last choice \emph{Ours}. We apply these strategies on $T_1$-weighted data using the Ridial sampling pattern with 20\% ratio. The stopping criterion for iterations is set as $\|\mathbf{x}^{k+1} - \mathbf{x}^{k}\|/\|\mathbf{x}^{k}\| \leq 1e-4$. 

	Figure~\ref{iteration} illustrates that the loss of $\mathcal{F}\!\!\to\!\!\mathcal{P}$ is slightly larger than that of $\mathcal{F}\!\!\to\!\!\mathcal{N}$ at the first several iterations. The data-driven module is more significant when image quality is low in the beginning. As the process goes on, the loss for $\mathcal{F}\!\!\to\!\!\mathcal{N}$ fluctuates and its PSNR decreases, showing unstable inference, while iterations for $\mathcal{F}\!\!\to\!\!\mathcal{P}$ keep stable. Combining deep and sparsity priors $\mathcal{F}\!\!\to\!\!\mathcal{N}\!\!\to\!\!\mathcal{P} $ improves performance over the former two, but converges not so stable and fast as \emph{Ours}. Higher PSNR values of \emph{Ours} demonstrate more accurate reconstruction. The execution time of \emph{Ours} is 2.5225s, less than that of $\mathcal{F}\!\!\to\!\!\mathcal{P}$ (4.4762s), $\mathcal{F}\!\!\to\!\!\mathcal{N}$ (3.3240s), and $\mathcal{F}\!\!\to\!\!\mathcal{N}\!\!\to\!\!\mathcal{P}$~(6.2760s). 
	Visual results in Fig.~\ref{component} also verify the superior quality of \emph{Ours} over the other three.
	
\subsection{Comparisons with the State-of-the-art}
	We compared our method with eleven CS-MRI techniques, including eight model based, Zero-Filling \cite{bernstein2001effect}, TV \cite{lustig2007sparse}, SIDWT \cite{baraniuk2007compressive}, PBDW \cite{qu2012undersled}, PANO \cite{qu2014magnetic}, FDLCP \cite{zhan2016fast}, BM3D-MRI \cite{eksioglu2016decoupled} and DAMP~\cite{eksioglu2018denoising}, as well as three learning based,  ADMM-Net~\cite{sun2016deep}, DAGAN~\cite{yang2017dagan} and RefineGAN~\cite{quan2018compressed}. The latter three methods and ours were tested on GPU.	
		
\textbf{Comparisons on benchmark sets:}
	We randomly selected 25 $T_1$-weighted and 25 $T_2$-weighted MRI data from 50 different subjects in the IXI dataset\footnote{http://brain-development.org/ixi-dataset/} for testing in order to evaluate reconstruction accuracy, time consumption and robustness to data variations. Three commonly used sampling patterns are adopted,~\emph{i.e.}, the Cartesian~\cite{qu2012undersled}, Radial~\cite{sun2016deep} and Gaussian~\cite{yang2017dagan}, with the sampling ratio varying from 10\% to 50\%. 
	
	The parameter $\rho$ in $\mathcal{F}$ is set as 5 and noise level in $\mathcal{N}$ ranges from 3.0 to 49.0. The parameters $L$ ($\eta=2/L$), $\lambda$ and $p$ in the modules $\mathcal{C}$ and $\mathcal{P}$ are set as $1.1$, $10^{-5}$ and $0.8$, respectively. The maximum number of iterations is 50. Comparative approaches take the parameter settings as their respective papers. We re-finetune deep models in DAGAN and RefineGAN upon their own using image pairs generated by various patterns and sampling ratios for fair comparisons.
	
	First, we evaluate reconstruction accuracy on benchmark data using three sampling patterns in terms of PSNR and RLNE. Table~\ref{tab:patterns} shows that ours leads all the competitors by a large margin for all sampling patterns. 
	Next, reconstruction efficiency is explored via comparisons on time consumption and reconstruction error with the Radial sampling  at a 50\% ratio. The marker at the lower left of the left plot in Fig.~\ref{fig:TimeRatio}  indicates that our method gives the least reconstruction error with a comparable low processing time, balancing accuracy and efficiency. 	
	Third, we use the Gaussian mask at five sampling ratios varying from 10\% to 50\%. As shown in the right plot of Fig.~\ref{fig:TimeRatio}, ours have the highest PSNR values over all others at all sampling ratios, demonstrating its robustness. 
	
	These comparisons demonstrates that traditional methods give ideal accuracy but suffer from large time consumption, while learning based ones enjoy fast forward inference but depend so highly on training data that testing data deviating from training distribution may severely degrade the performance.
	
	\textbf{Comparisons on Real Complex Data:}
	 We conduct comparisons on real $k$-space data directly collected from a Philip machine to investigate the performance for practical applications. Many algorithms like BM3D-MRI and those learning based ones consider no complex form of these real data. The flexibility of our framework enables us to directly deal with complex data by the modules $\mathcal{F,C}$ and $\mathcal{P}$. We first separately handle the real and imaginary parts, and then merge to the  complex form for the module $\mathcal{N}$. In this experiment, the noise level $\mathcal{N}$ ranges from 20.0 to 49.0. The parameter $L$ in $\mathcal{C}$ and $\mathcal{P}$ is set to 1.7, and the maximum iteration number is 60. 
	
	The PSNR, SSIM and RLNE values of our method in Tab.~\ref{tab:RealData} demonstrate more accurate reconstruction than the others on real data. Figure~\ref{fig:RealData} gives the qualitative comparisons with enlarged details and error heat maps. Our method produces sharper textures and less error, showing the superior capability.

	\subsection{Experiments on Parallel Imaging}
	We perform reconstruction from multi-coils by comparing with four PI techniques, SENSE~\cite{pruessmann1999sense}, GRAPPA~\cite{griswold2005parallel}, SPIRiT~\cite{lustig2010spirit} and $L_1$-SPIRiT~\cite{murphy2012fast}. 
	Table~\ref{tab:parallelImaging} shows that our approach achieves higher accuracy in less time, possessing an ideal ability to handle multi-coil data. The top row of Fig.~\ref{fig:parallelImaging} visualizes the input from multiple coils, and the lower row shows the reconstruction using different techniques, where our method restores richer details with higher PSNR.
	
	\subsection{Experiments on reconstruction with Rician noise}		
	The two-stage deep networks solving the subproblem~\eqref{eq:rician_sub_x} acts as a Rician denoiser. We first compare the strategy with classical Rician denoisers including NLM~\cite{buades2005review}, UNLM~\cite{manjon2008mri}, LMMSE~\cite{aja2008restoration} and RiceVST~\cite{foi2011noise}. We generate training data by adding Rician noise at different levels to 500 $T_1$-weighted MRI data randomly chosen from the MICCAI 2013 grand challenge dataset\footnote{http://masiweb.vuse.vanderbilt.edu/workshop2013/index.php\\/Segmentation\_Challenge\_Details}.		
	Figure~\ref{fig:NoisyError} shows error heatmaps between the ground truth and denoised image by various techniques. Our strategy successfully resolves the non-additivity so that it attenuates noise more effectively than the others, especially for background. RiceVST performs close to ours, and hence we further compare with RiceVST on five noise levels varying from 10 to 50, shown in Fig.~\ref{fig:NoisySigma}. Ours outputs higher PSNR at levels less than 30 with much less time than RiceVST.
	
	Our framework naturally embeds the two-stage deep denoiser into iterations for reconstruction. We investigate the performance of this one-set solution for reconstruction with Rician noise on $T_1$-weighted data in the IXI dataset. The parameters $\rho_1$, $\lambda_1$ and  $\lambda_2$ in~\eqref{eq:rician_sub_z} and~\eqref{eq:rician_sub_x} are set as 0.01, 1.0 and 1.0, respectively. For fair comparisons, we cascade RiceVST or our two-stage denoiser to the other reconstruction algorithms as post-processing. The right plot of Fig.~\ref{fig:NoisySigma} illustrates PSNR values reconstructed by the comparative algorithms with RiceVST ($x$-axis) and our two-stage denoiser ($y$-axis) cascaded. The brown triangle denotes PSNR ($28$dB) obtained by our one-set solution. This figure shows that a) our framework combing reconstruction with denoising produces the best quality over the other reconstruction algorithms with post-processing, and b) our learnable two-stage denoiser performs better than RiceVST as post-pocessing for reconstruction. Figure~\ref{fig:noisyDetail} provides an visual illustration for reconstruction with noise, where our framework also preserves more details and introduces lower errors.

	\section{Conclusion}
	We propose a theoretically converged deep framework for CS-MRI reconstruction. This paradigm takes advantages of both model-based and deep representation, improving accuracy and efficiency. Further, we devise an optimal condition that guides iterative propagation converging to the critical point of the energy with priors, guaranteeing reliability. We also apply this framework to parallel imaging and reconstruction with Rician noise for practical scenarios, without needing significant data re-train. Extensive experiments demonstrate the superiority of our framework over the state-of-the-art. Future work directs to applications for more tasks like low-dose CT, and to incorporation of other effective deep architectures. 
	
	\section*{Acknowledgments}	
	This paper extends from the conference version published at AAAI 2019~\cite{liu2019theoretically}. The authors would like to thank MR method manager Drs. Chenguang Zhao and sequence engineer Yuli Huang at Philips Healthcare (Suzhou) Co. Ltd. for providing testing data from Prodiva 1.5T scanner (Philips Healthcare, Best, Netherlands). This work is partially supported by the National Natural Science Foundation of China (Nos. 61922019, 61672125, 61572096 and 61632019), and the Fundamental Research Funds for the Central Universities.

	\section*{Supplemental Materials}\label{sec:proofs}
	\subsection*{Proof}
	To simplify the following derivations, we first rewrite the function in Eq.~\eqref{eq:OriginalModel} as
	$$
	\begin{array}{l}
	\Phi (\bm{\alpha}) = f(\bm{\alpha}) + g(\bm{\alpha}) =\frac{1}{2}\| \mathbf{PF\mathbf{A}\bm{\alpha}-y}\|_2^2+\lambda \|\bm{\alpha}\|_p.
	\end{array}
	$$
	We first provide detailed explanations about the properties of $f, g$, and $\Phi$.
	
	\begin{itemize}
		\item $f(\bm{\alpha})$ is proper if $\mathtt{dom}f:=\{\bm{\alpha}\in\mathbb{R}^D: f(\bm{\alpha})<+\infty\}$ is nonempty.
		\item $f(\bm{\alpha})$ is $L$-Lipschitz smooth if $f$ is differentiable and there exists $L>0$ such that 
		$$
		\|\nabla f(\bm{\alpha}) - \nabla f(\bm{\beta})\| \leq L \|\bm{\alpha} - \bm{\beta}\|, \ \forall \ \bm{\alpha},\bm{\beta} \in \mathbb{R}^{D}.
		$$
		If f is $L$-Lipschitz smooth, we have the following inequality
		$$
		f(\bm{\alpha})\!\leq\! f(\bm{\beta}) + \langle \nabla f(\bm{\beta}),\! \bm{\alpha}-\bm{\beta}\rangle + \frac{L}{2}\|\bm{\alpha}-\bm{\beta}\|^2, \ \forall \bm{\alpha}, \bm{\beta}\!\in\!\mathbb{R}^D.
		$$
		\item $g(\bm{\alpha})$ is lower semi-continuous if $\liminf \limits_{\bm{\alpha}\to\bm{\beta}}g(\bm{\alpha})\geq g(\bm{\beta})$ at any point $\bm{\beta}\in\mathtt{dom}g$.
		\item $\Phi(\bm{\alpha})$ is coercive if $\Phi$ is bounded from below and $\Phi\to\infty$ if $\|\bm{\alpha}\|\to\infty$, where $\|\cdot\|$ is the $\ell_2$ norm.
		\item $\Phi(\bm{\alpha})$ is is said to have the Kurdyka-{\L}ojasiewicz (K\L) property at $\bar{\bm{\alpha}}\in\mathtt{dom}\partial \Phi:=\{\bm{\alpha}\in\mathbb{R}^D: \partial g(\bm{\alpha}) \neq \emptyset\}$ if there exist $\eta\in(0,\infty]$, a neighborhood $\mathcal{U}_{\bar{\bm{\alpha}}}$ of $\bar{\bm{\alpha}}$ and a desingularizing function $\phi:[0,\eta)\to \mathbb{R}_+$ which satisfies (1) $\phi$ is continuous at $0$ and $\phi(0)=0$; (2) $\phi$ is concave and $C^1$ on $(0,\eta)$; (3) for all $s\in(0,\eta): \phi'(s)>0$, such that for all
		$$
		\bm{\alpha}\in\mathcal{U}_{\bar{\bm{\alpha}}}\cap[\Phi(\bar{\bm{\alpha}})<\Phi(\bm{\alpha})<\Phi(\bar{\bm{\alpha}})+\eta],
		$$
		the following inequality holds
		$$
		\phi'(\Phi(\bm{\alpha})-\Phi(\bar{\bm{\alpha}}))\mathtt{dist}(0,\partial \Phi(\bm{\alpha})) \geq 1.
		$$
		Moreover, if $\Phi$ satisfies the K{\L} property at each point of $\mathtt{dom}\partial \Phi$ then $\Phi$ is called a K{\L} function.
	\end{itemize}
	\begin{proposition}\label{prop:c-error}
		Let $ \left\{\bm{\alpha}^k\right\}_{k\in\mathbb{N}} $ and  $\left\{\bm{\beta}^k\right\}_{k\in\mathbb{N}} $ be the sequences generated by Alg.\ref{alg1}. Suppose that the error condition 
		$\|\mathbf{v}^{k+1}-\bm{\alpha}^{k}\| \leq \varepsilon^{k}\|\bm{\beta}^{k+1}-\bm{\alpha}^{k}\|$ 
		in our $\mathtt{icheck}$ is satisfied. Then there existed a sequence $\{C^{k}\}_{k\in\mathbb{N}}$ such that 
		\begin{equation}	
		\Phi(\bm{\beta}^{k+1}) \leq \Phi(\bm{\alpha}^k)-C^{k}\|\bm{\beta}^{k+1}-\bm{\alpha}^k\|^2,
		\end{equation}
		where $C^{k} = \frac{1}{2\eta_{1}} - \frac{L_f}{2} -  (L_f  + |\rho-\frac{1}{\eta_{1}}|)\epsilon^{k} >0$ and $L_f$ is the Lipschitz coefficient of $\nabla f$ .
		\label{eq:ineq_fun_pgmomentum}
	\end{proposition}

	\begin{proof}
		In our $\mathtt{icheck}$ stage, we have 
		$$
		\bm{\beta}^{k+1}\!\in\!\emph{prox}_{\eta_1 \lambda\|\cdot\|_p}
		\!\left(\mathbf{v}^{k+1}\!-\!\eta_1\nabla f\left(\mathbf{v}^{k+1}\right)\!+	
		\!\rho\left(\mathbf{v}^{k+1}\!-\!\bm{\alpha}^{k}\right)\right),
		$$
		if the error condition is satisfied.
		Thus, by the definition of the proximal operator we get 
		\begin{equation}
		\begin{aligned}
		\bm{\beta}^{k+1} 
		&\!\in\!\arg\min\limits_{\bm{\beta}} \eta_{1}\lambda\|\bm{\beta} \|_{p}\! +\!\frac{1}{2}\|\bm{\beta}\! - \!\mathbf{v}^{k+1}\!+\!\eta_1 \left( \nabla f\left(\mathbf{v}^{k+1}\right)\right.\\
		&\left.\quad +\rho\left(\mathbf{v}^{k+1}\!-\bm{\alpha}^{k}\right) \right) \|^2 \\
		&= \arg\min\limits_{\bm{\beta}} \eta_{1}\lambda\|\bm{\beta} \|_{p} \!+
		\!\frac{1}{2}\|\bm{\beta} \!-\!\bm{\alpha}^{k}\!+\! \eta_1\nabla f\left(\mathbf{v}^{k+1}\right)	\\
		&\quad \  +(\eta_1\rho-1)\left(\mathbf{v}^{k+1}-\bm{\alpha}^{k}\right) \|^2\\
		&=\arg\min\limits_{\bm{\beta}} \lambda\|\bm{\beta} \|_{p} + \frac{1}{2\eta_{1}}\|\bm{\beta} -\bm{\alpha}^{k} \|^2 +	\\ &\quad \left\langle\bm{\beta} -	\bm{\alpha}^{k}, \nabla f\left(\mathbf{v}^{k+1}\right)+
		(\rho-\frac{1}{\eta_{1}})\left(\mathbf{v}^{k+1}-\bm{\alpha}^{k}\right)  \right\rangle.
		\end{aligned}
		\label{eq:min_func_pgmoment}
		\end{equation}
		Hence in particular, taking $\bm{\beta} = \bm{\beta}^{k+1}$ and $\bm{\beta} = \bm{\alpha}^{k}$ respectively, we obtain
		\begin{equation}
		\begin{array}{l}
		\lambda\|\bm{\beta}^{k+1}\|_{p} + \frac{1}{2\eta_{1}}\|\bm{\beta}^{k+1} -\bm{\alpha}^{k} \|^2 + \left\langle\bm{\beta}^{k+1} -\bm{\alpha}^{k}, \right.	\\
		\nabla f\left(\mathbf{v}^{k+1}\right)
		\left.+(\rho-\frac{1}{\eta_{1}})\left(\mathbf{v}^{k+1}-\bm{\alpha}^{k}\right)  \right\rangle 
		\leq \lambda\|\bm{\alpha}^{k}\|_{p}.
		\end{array}\label{eq:ineq_g}
		\end{equation}
		Invoking the Lipschitz smooth property of $f$, we also have
		\begin{equation}
		f(\bm{\beta}^{k+1})\! \leq\! f(\bm{\alpha}^{k})\!  +  \!\left\langle\bm{\beta}^{k+1}\!-\!\bm{\alpha}^{k},\! \nabla f(\bm{\alpha}^{k}) \right\rangle\! +\! \frac{L_f}{2}\| \bm{\beta}^{k+1} \!-\!\bm{\alpha}^{k}\|^2.\label{eq:ineq_f}
		\end{equation} 
		Combing Eqs.~\eqref{eq:ineq_g} and \eqref{eq:ineq_f}, we obtain
		$$
		\begin{array}{l}
		\begin{aligned}
		\!\!\!\Phi(\bm{\beta}^{k+1}) &\leq \Phi(\bm{\alpha}^{k})\! - \!\frac{1}{2\eta_{1}}\|\bm{\beta}^{k+1}\! -\!\bm{\alpha}^{k} \|^2 \!+\! \frac{L_f}{2}\| \bm{\beta}^{k+1} \!-\!\bm{\alpha}^{k}\|^2\\
		&\quad+ \left\langle \bm{\beta}^{k+1} -\bm{\alpha}^{k}, \nabla f(\bm{\alpha}^{k})  - \nabla f\left(\mathbf{v}^{k+1}\right)	\right.\\
		&\quad\left.-(\rho-\frac{1}{\eta_{1}})\left(\mathbf{v}^{k+1}-\bm{\alpha}^{k}\right)\right\rangle \\
		&\leq \Phi(\bm{\alpha}^{k})\! -\!  \frac{1}{2\eta_{1}}\|\bm{\beta}^{k+1}\!  -\! \bm{\alpha}^{k} \|^2 \! +\!  \frac{L_f}{2}\| \bm{\beta}^{k+1} \! -\! \bm{\alpha}^{k}\|^2 \\
		&\quad+(L_f  + |\rho-\frac{1}{\eta_{1}}|)\epsilon^{k} \|\bm{\beta}^{k+1} -\bm{\alpha}^{k} \|^2 \\
		&\leq \Phi(\bm{\alpha}^{k})\!  -\! (\frac{1}{2\eta_{1}}\!  - \! \frac{L_f}{2}\!  - \!  (L_f \!  +\!  |\rho\! -\! \frac{1}{\eta_{1}}|)\epsilon^{k})  \|\bm{\beta}^{k+1} \\
		&\quad-\bm{\alpha}^{k} \|^2 \\
		&\leq \Phi(\bm{\alpha}^{k}) -C^{k}\|\bm{\beta}^{k+1} -\bm{\alpha}^{k} \|^2.
		\end{aligned}
		\end{array}
		$$
		The last inequality holds under the 
		assumption $C^{k} = \frac{1}{2\eta_{1}} - \frac{L_f}{2} -  (L_f  + |\rho-\frac{1}{\eta_{1}}|)\epsilon^{k}>0$ in the $k$-th iteration. So far, this prove that the inequality ~\eqref{eq:ineq_fun_pgmomentum} in Proposition~\ref{prop:c-error} holds.
	\end{proof}
	
	\begin{proposition}\label{prop:pg}
		If $\eta_2 < 1/L_f$, let $ \left\{\bm{\alpha}^k\right\}_{k\in\mathbb{N}} $ and $\left\{\mathbf{w}^k\right\}_{k\in\mathbb{N}}$ be the sequences generated by a proximal operator in Alg.\ref{alg1}. Then we have
		\begin{equation}
		\Phi(\bm{\alpha}^{k+1}) \leq \Phi(\mathbf{w}^{k+1})-({1}/({2\eta_2}) - {L_f}/{2})\|\bm{\alpha}^{k+1}-\mathbf{w}^{k+1}\|^2.\label{eq:ineq_fun_pg}
		\end{equation}	
	\end{proposition}
	
	\begin{proof}
		As the proximal stage shows in Alg.~\ref{alg1}, we have 
		\begin{equation}
		\begin{split}
		\bm{\alpha}^{k+1}  &\in \emph{prox}_{\eta_2 \lambda\|\cdot\|_p}\left(\mathbf{w}^{k+1}-\eta_2\nabla f\left(\mathbf{w}^{k+1}\right)\right)\\
		&= \arg\min\limits_{\bm{\alpha}} \lambda \| \bm{\alpha}\|_{p} + \frac{1}{2\eta_2}\| \bm{\alpha} -\mathbf{w}^{k+1} \|^2 +\\
		&\quad \langle\bm{\alpha} -\mathbf{w}^{k+1} ,\nabla f(\mathbf{w}^{k+1}) \rangle
		\label{eq:min_func_pg}.
		\end{split}		
		\end{equation}
		Similar with the deduction in Proposition~\ref{prop:c-error}, we obtain the following inequality:
		$$
		\begin{array}{l}
		\lambda\| \bm{\alpha}^{k+1}\|_{p} + \frac{1}{2\eta_{2}}\|\bm{\alpha}^{k+1} -\mathbf{w}^{k+1}  \|^{2} + 
		\left\langle \bm{\alpha}^{k+1} -\mathbf{w}^{k+1}  , \right.\\
		\left.\nabla f(\mathbf{w}^{k}) \right\rangle \leq \lambda\| \mathbf{w}^{k+1}\|_{p}, \\\\
		f(\bm{\alpha}^{k+1}) \leq f(\mathbf{w}^{k+1}) + \left\langle \bm{\alpha}^{k+1} -\mathbf{w}^{k+1}, \nabla f(\mathbf{w}^{k+1}) \right\rangle  \\
		+\frac{L_f}{2}\| \bm{\alpha}^{k+1} -\mathbf{w}^{k+1}\|^2.
		\end{array}
		$$
		Thus we get the conclusion
		$$
		\Phi(\bm{\alpha}^{k+1}) \leq \Phi(\mathbf{w}^{k+1})-(\frac{1}{2\eta_2} - \frac{L_f}{2})\|\bm{\alpha}^{k+1}-\mathbf{w}^{k+1}\|^2.
		$$
		
	\end{proof}

	\begin{theorem}
		Suppose that $ \left\{\bm{\alpha}^k\right\}_{k\in\mathbb{N}} $ be a sequence generated by Alg.~\ref{alg1}. The following assertions hold.
		\begin{itemize}
			\item The square summable of sequence $\left\{\bm{\alpha}^{k+1}-\mathbf{w}^{k+1} \right\}_{k\in\mathbb{N}}$ is bounded, i.e., 
			$$
			\sum_{k=1}^{\infty}\|\bm{\alpha}^{k+1}-\mathbf{w}^{k+1}\|^2 < \infty.
			$$
			\item The sequence $ \left\{\bm{\alpha}^k\right\}_{k\in\mathbb{N}} $ converges to a critical point $\bm{\alpha}^{*}$ of $\Phi$.
		\end{itemize}
	\end{theorem}
	
	\begin{proof}
		We first verify that square summable of $\left\{\bm{\alpha}^k -\bm{w}^{k} \right\}_{k\in\mathbb{N}}$ is bounded. From Propositions~\ref{prop:c-error} and \ref{prop:pg}, we deduce the 
		$$
		\Phi(\bm{\alpha}^{k+1}) \leq \Phi(\mathbf{w}^{k+1}) \leq \Phi(\bm{\alpha}^{k})  \leq \Phi(\mathbf{w}^{k}) \leq \Phi(\bm{\alpha}^{0}),
		$$
		is established. It follows that both sequences $\left\{ \Phi(\bm{\alpha}^{k}) \right\}_{k\in\mathbb{N}}$ and $\left\{\Phi(\mathbf{w}^{k})\right\}_{k\in\mathbb{N}}$ are non-increasing. Then, since both $f$ and $g$ are proper, we have $\Phi$ is bounded and that the objective sequences $\left\{ \Phi(\bm{\alpha}^{k}) \right\}_{k\in\mathbb{N}}$ and $\left\{\Phi(\mathbf{w}^{k})\right\}_{k\in\mathbb{N}}$ converge to the same value $\Phi^{*}$, i.e.,
		$$
		\lim\limits_{k \to \infty} \Phi (\bm{\alpha}^{k} ) = \lim\limits_{k \to \infty} \Phi (\mathbf{w}^{k} ) = \Phi^{*}.
		$$
		Moreover, using the assumption $\Phi$ is coercive, we have that both $\left\{\bm{\alpha}^k \right\}_{k\in\mathbb{N}}$ and $\left\{\mathbf{w}^{k} \right\}_{k\in\mathbb{N}}$ are bounded and thus have accumulation points.
		Considering Eq.~\eqref{eq:ineq_fun_pg} and the relationship of $\Phi(\bm{\alpha})$ and $\Phi(\mathbf{w})$, we get for any $k\geq0$,
		\begin{equation}
		\begin{split}
		(\frac{1}{2\eta_{2}}-\frac{L_f}{2})\|\bm{\alpha}^{k+1}-\mathbf{w}^{k+1}\|^2 &\leq \Phi(\mathbf{w}^{k+1}) - \Phi(\bm{\alpha}^{k+1}) \\
		&\leq \Phi(\bm{\alpha}^{k}) -\Phi(\bm{\alpha}^{k+1}). \label{eq:phi_k_k+1}
		\end{split}
		\end{equation}
		Summing over $k$, we further have
		\begin{equation}
		(\frac{1}{2\eta_{2}}-\frac{L_f}{2}) \sum\limits_{k=0}^{\infty} \|\bm{\alpha}^{k+1}-\mathbf{w}^{k+1}\|^{2} \leq \Phi(\bm{\alpha}^{0}) -\Phi^{*}<\infty.\label{eq:ineq_phi}
		\end{equation}
		So far, the first assertion holds.
		
		Next, we prove $\lim\limits_{k \to \infty} \bm{\alpha}^{k} = \bm{\alpha}^{*}$ is a critical point of $\Phi$. Eq.~\eqref{eq:ineq_phi} implies that $\|\bm{\alpha}^{k+1}-\mathbf{w}^{k+1}\| \to 0$ when $k \to \infty$, i.e., there exist subsequences $\{\mathbf{w}^{k} \}_{k_j\in\mathbb{N}}$ and $\{\bm{\alpha}^{k} \}_{k_j\in\mathbb{N}}$ such that they share the same accumulation point $\bm{\alpha}^{*}$ as $j\to \infty$.
		Then by Eq.~\eqref{eq:min_func_pg}, we have 
		$$
		\begin{aligned}		
		&\lambda\| \bm{\alpha}^{k+1}\|_{p} \! +\!  \frac{1}{2\eta_{2}}\|\bm{\alpha}^{k+1}\!  -\! \mathbf{w}^{k+1}  \|^{2} \! +\!  \left\langle \bm{\alpha}^{k+1} \! -\! \mathbf{w}^{k+1} \!  ,\!  \nabla f(\mathbf{w}^{k}) \right\rangle \\	
		&\leq 
		\lambda\| \bm{\alpha}^{*}\|_{p} + \frac{1}{2\eta_{2}}\|\bm{\alpha}^{*} -\mathbf{w}^{k+1}  \|^{2} + \left\langle \bm{\alpha}^{*} -\mathbf{w}^{k+1}  , \nabla f(\mathbf{w}^{k}) \right\rangle.
		\end{aligned}
		$$
		Let $k+1 = k_j$ and $j \to \infty$,  then by taking $\lim\sup$ on the above inequality, we have 
		$$
		\lim\sup\limits_{j\to \infty} \| \bm{\alpha}^{k_j}\|_{p} \leq \| \bm{\alpha}^{*}\|_{p}.
		$$
		What's more, we also get $ \lim\inf\limits_{j\to \infty} \| \bm{\alpha}^{k_j}\|_{p} \geq \| \bm{\alpha}^{*}\|_{p}$ since $\|\cdot\|_{p}$ is lower semi-continuous. Thus we have $\lim\limits_{j\to \infty} \| \bm{\alpha}^{k_j}\|_{p} = \| \bm{\alpha}^{*}\|_{p}$.
		Considering the continuity of $f$, we have $\lim\limits_{j\to \infty} f(\bm{\alpha}^{k_j}) = f(\bm{\alpha}^{*}). $ Thus, we obtain
		\begin{equation}
		\lim\limits_{j\to \infty}\Phi(\bm{\alpha}^{k_j}) = \lim\limits_{j\to \infty} f(\bm{\alpha}^{k_j}) + \lambda \|\bm{\alpha}^{k_j}\|_{p} = \Phi(\bm{\alpha}^{*}).
		\label{eq:Phi_lim}
		\end{equation}
		By the first-order optimality condition of Eq.~\eqref{eq:min_func_pg} and $k_j = k+1$, we deduce
		$$
		0\in \partial \lambda\|\bm{\alpha}^{k_j}\|_p + \nabla f(\mathbf{w}^{k_j}) + \frac{1}{\eta_{2}} ( \bm{\alpha}^{k_j}-\mathbf{w}^{k_j}).
		$$
		Hence we get
		\begin{equation}
		\begin{array}{l}
		\quad \nabla f(\bm{\alpha}^{k_j})-\nabla f(\mathbf{w}^{k_j}) - \frac{1}{\eta_{2}} ( \bm{\alpha}^{k_j}-\mathbf{w}^{k_j}) \in \partial \Phi(\bm{\alpha}^{k_j}) 
		\Rightarrow \\\| \partial \Phi(\bm{\alpha}^{k_j})\| = \|\nabla f(\bm{\alpha}^{k_j})-\nabla f(\mathbf{w}^{k_j}) - \frac{1}{\eta_{2}} ( \bm{\alpha}^{k_j}-\mathbf{w}^{k_j}) \| \\
		\qquad \qquad  \quad  \leq (L_f + \frac{1}{\eta_{2}}) \|\bm{\alpha}^{k_j}-\mathbf{w}^{k_j}\|.
		\end{array}\label{eq:sub-diff_phi}
		\end{equation}
		Using the sub-differential of $\Phi$ and Eqs.~\eqref{eq:Phi_lim}, ~\eqref{eq:sub-diff_phi}, we finally deduce that $0 \in \Phi(\bm{\alpha}^{*})$, which means that $\{\bm{\alpha}^{k}\}_{k \in \mathbb{N}}$ is subsequence convergence.
		
		Furthermore, we will prove that $\{\bm{\alpha}^{k}\}_{k \in \mathbb{N}}$ is sequence convergence. Since $\Phi$ is a K{\L} function, we have
		$$
		\varphi'(\Phi(\bm{\alpha}^{k+1}) -\Phi(\bm{\alpha}^*) ) \mathtt{dist}(0, \partial \Phi(\bm{\alpha}^{k+1})) \geq 1.
		$$
		From Eq.~\eqref{eq:sub-diff_phi} we get that
		$$
		\varphi'(\Phi(\bm{\alpha}^{k+1}) -\Phi(\bm{\alpha}^*) ) \geq \frac{1}{L_f + \frac{1}{\eta_2}}\| \bm{\alpha}^{k_j}-\mathbf{w}^{k_j} \|^{-1}.
		$$
		On the other hand, from the concavity of $\varphi$ and Eq.~\eqref{eq:phi_k_k+1} and ~\eqref{eq:sub-diff_phi}
		we have that 
		$$
		\begin{array}{l}
		\varphi(\Phi(\bm{\alpha}^{k+1}) -\Phi(\bm{\alpha}^*) ) -
		\varphi(\Phi(\bm{\alpha}^{k+2}) -\Phi(\bm{\alpha}^*) ) \\
		\geq \varphi'(\Phi(\bm{\alpha}^{k+1}) -\Phi(\bm{\alpha}^*) ) (\Phi(\bm{\alpha}^{k+1}) - \Phi(\bm{\alpha}^{k+2}) ) \\
		\geq \frac{1}{L_f + \frac{1}{\eta_2}}\| \bm{\alpha}^{k+1}\! -\! \mathbf{w}^{k+1} \|^{-1}
		(\frac{1}{2\eta_2} \! - \! \frac{L_f}{2})
		\|\bm{\alpha}^{k+2} \! -\!  \mathbf{w}^{k+2}\|^{2}.
		\end{array}
		$$
		For convenience, we define for all $m,n \in \mathbb{N}$ and $\bm{\alpha}^{*}$ the following quantities
		$$
		\Delta_{m,n} := \varphi(\Phi(\bm{\alpha}^{m}) -\Phi(\bm{\alpha}^*) ) -
		\varphi(\Phi(\bm{\alpha}^{n}) -\Phi(\bm{\alpha}^*) ) ,
		$$
		and 
		$$
		E:= \frac{2L_f \eta_2+2}{1-L_f \eta_2}.
		$$
		These deduce that 
		$$
		\begin{array}{l}
		\quad \Delta_{k+1,K+2} \geq \frac{\|\bm{\alpha}^{k+2} - \mathbf{w}^{k+2}\|^2}{E \|\bm{\alpha}^{k+1}-\mathbf{w}^{k+1}\| }\\
		\Rightarrow \|\bm{\alpha}^{k+2} - \mathbf{w}^{k+2}\|^2 \leq E \Delta_{k+1,k+2} \|\bm{\alpha}^{k+1}-\mathbf{w}^{k+1}\|\\
		\Rightarrow 2 \|\bm{\alpha}^{k+2} \! -\!  \mathbf{w}^{k+2} \| \leq E \Delta_{k+1,k+2}\!  +\!  \|\bm{\alpha}^{k+1}\! -\! \mathbf{w}^{k+1}\|.
		\end{array}
		$$	
		Summing up the above inequality for $i=l,\dots,k$ yields
		$$
		\begin{aligned}
		&\sum\limits_{i=l+1}^{k} 2 \|\bm{\alpha}^{i+2} -\mathbf{w}^{i+2} \|\\
		&\leq \!\sum\limits_{i=l+1}^{k} \!\| \bm{\alpha}^{i+1}-\mathbf{w}^{i+1} \|+E \sum\limits_{i=l+1}^{k} \Delta_{i+1,i+2} \\ 
		&\leq \!\sum\limits_{i=l+1}^{k} \| \bm{\alpha}^{i+2}-\mathbf{w}^{i+2} \| + \|\bm{\alpha}^{l+2}-\mathbf{w}^{l+2} \| + E \Delta_{l+2,k+2},
		\end{aligned}		
		$$
		where the last inequality holds under the fact $\Delta_{m,n}+\Delta_{n,r} = \Delta_{m,r}$ for all $m,n,r \in \mathbb{N}$.
		Since $\varphi>0$, we thus have for any $k>l$ that
		\begin{equation}
		\begin{aligned}
		&\sum\limits_{i=l+1}^{k}  \|\bm{\alpha}^{i+2} - \mathbf{w}^{i+2} \|
		\leq \|\bm{\alpha}^{l+2}-\mathbf{w}^{l+2} \| + E \Delta_{l+2,k+2} \\
		&\leq \|\bm{\alpha}^{l+2}-\mathbf{w}^{l+2} \| + E \varphi(\Phi(\bm{\alpha}^{l+2}) - \Phi(\bm{\alpha}^{*})).
		\label{eq:a2w2}
		\end{aligned}
		\end{equation}
		Moreover, recalling the conclusion in Proposition~\ref{prop:c-error},  we also has
		\begin{equation}
		\begin{split}
		&\!\min\limits_{i}\{C^{i}\}\! \sum\limits_{i=l+1}^{k} \!\| \mathbf{w}^{i+2} \!-\! \bm{\alpha}^{i+1}\|\! \leq\!c \sum\limits_{i=l+1}^{k}\! (\Phi(\mathbf{w}^{i+2})\! -\!\Phi(\bm{\alpha}^{i+1}\!) ) \\
		&\leq  \sum\limits_{i=l+1}^{k} (\Phi(\bm{\alpha}^{i+2}) -\Phi(\bm{\alpha}^{i+1}) ) = \Phi(\bm{\alpha}^{k+2}) - \Phi(\bm{\alpha}^{l+2}).
		\label{eq:w2a1}
		\end{split}
		\end{equation}
		Combing with Eqs.~\eqref{eq:a2w2} and \eqref{eq:w2a1}, we easily deduce
		\begin{equation}
		\sum\limits_{k=1}^{\infty} \|\bm{\alpha}^{k+1} - \bm{\alpha}^{k} \| < \infty.
		\label{eq:cauchy}
		\end{equation}
		It is clear that Eq.~\eqref{eq:cauchy} implies that the sequence $\{ \bm{\alpha}^{k}\}_{k \in \mathbb{N}}$ is a Cauchy sequence and hence is a convergent sequence. 
		So far, the second assertion holds.
		
	\end{proof}
	
	\subsection*{Sampling Masks}
	In experiments we include three common used types of undersampling masks such as the Cartesian pattern in \cite{qu2012undersled}, Radial pattern in \cite{sun2016deep}  and Gaussian mask in \cite{yang2017dagan}. Fig.\ref{mask} gives a visualization of the three kinds of patterns at a unified sampling ratio of 30\%.
	\begin{figure}[!htbp]
		\begin{center}
			\begin{tabular}{c@{\extracolsep{0.8em}}c@{\extracolsep{0.8em}}c}
				\includegraphics[width=.14\textwidth]{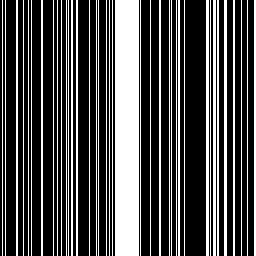}
				&\includegraphics[width=.14\textwidth]{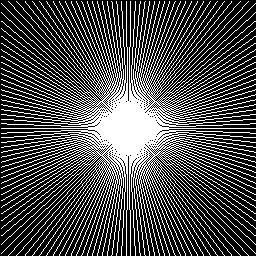}
				&\includegraphics[width=.14\textwidth]{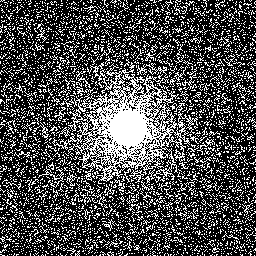}\\
				Cartesian & Radial & Gaussian
			\end{tabular}
		\end{center}
		\caption{Three types of sampling pattern.}
		\label{mask}
	\end{figure}	
	\ifCLASSOPTIONcaptionsoff
	\newpage
	\fi

	\bibliographystyle{IEEEtran}
	\bibliography{IEEEabrv,reference}

\end{document}